\documentclass[fullpage,11pt]{article}

\usepackage{amsthm}
\usepackage{graphicx} 
\usepackage{array} 
\usepackage{color}
\usepackage{amsmath, amssymb, amsfonts, verbatim}
\usepackage{hyphenat,epsfig,subfigure,multirow}

\usepackage[usenames,dvipsnames]{xcolor}
\usepackage[ruled]{algorithm2e}

\usepackage{tcolorbox}
\tcbuselibrary{skins,breakable}
\tcbset{enhanced jigsaw}
 
\definecolor{DarkRed}{rgb}{0.5,0.1,0.1}
\definecolor{DarkBlue}{rgb}{0.1,0.1,0.5}

\usepackage{hyperref}
\hypersetup{
colorlinks=true,
pdfnewwindow=true,
citecolor=ForestGreen,
linkcolor=DarkRed,
filecolor=DarkRed,
urlcolor=DarkBlue
}

\usepackage[compact]{titlesec}

\usepackage{bm}
\usepackage{url}
\usepackage{xspace} 
\usepackage[mathscr]{euscript}

\usepackage{mdframed}

\usepackage[noend]{algpseudocode}
\makeatletter
\def\BState{\State\hskip-\ALG@thistlm}
\makeatother

\usepackage{cite}
\usepackage{enumerate}

\usepackage[margin=1in]{geometry}

\newtheorem{theorem}{Theorem}
\newtheorem{lemma}{Lemma}[section]
\newtheorem{proposition}[lemma]{Proposition}
\newtheorem{corollary}[theorem]{Corollary}
\newtheorem{claim}[lemma]{Claim}
\newtheorem{fact}[lemma]{Fact}

\newtheorem{definition}[lemma]{Definition}

\newtheorem{problem}{Problem}
\newtheorem{remark}[lemma]{Remark}
\newtheorem{observation}[lemma]{Observation}

\newtheorem*{claim*}{Claim}
\newtheorem*{proposition*}{Proposition}
\newtheorem*{lemma*}{Lemma}
\newtheorem*{problem*}{Problem}
\newtheorem*{definition*}{Definition}

\newtheorem{mdresult}{Result}
\newenvironment{result}{\begin{mdframed}[backgroundcolor=lightgray!40,topline=false,rightline=false,leftline=false,bottomline=false,innertopmargin=2pt]\begin{mdresult}}{\end{mdresult}\end{mdframed}}

\allowdisplaybreaks

\DeclareMathOperator*{\argmax}{arg\,max}

\newcommand{\mi}[2]{\ensuremath{\II(#1 \,; #2)}}
\newcommand{\en}[1]{\ensuremath{\HH(#1)}}

\renewcommand{\qed}{\nobreak \ifvmode \relax \else
      \ifdim\lastskip<1.5em \hskip-\lastskip
      \hskip1.5em plus0em minus0.5em \fi \nobreak
      \vrule height0.75em width0.5em depth0.25em\fi}

\usepackage[T1]{fontenc}
\usepackage[utf8]{inputenc}
\usepackage{textcomp}

\usepackage[title]{appendix}

\newcommand{\ourinfo}[1]{Department of Computer and Information Science, University of Pennsylvania. Supported in part by National Science Foundation grants CCF-1552909 and CCF-1617851.  \noindent Email: \texttt{#1}.}

\newcommand{\toShrink}{-.20cm}
\newcommand{\toShrinkEnu}{-.2cm}

\newenvironment{tbox}{\begin{tcolorbox}[
		enlarge top by=5pt,
		enlarge bottom by=5pt,
		breakable,%
		 boxsep=0pt,
                  left=4pt,
                  right=4pt,
                  top=10pt,
                  arc=0pt,
                  boxrule=1pt,toprule=1pt,
                  colback=white
                  ]
	}
{\end{tcolorbox}}

\newcommand{\textbox}[2]{
{
\begin{tbox}
\textbf{#1}
{#2}
\end{tbox}
}
}

\newcommand{\itfacts}[1]{Fact~\ref{fact:it-facts}-(\ref{part:#1})\xspace}
\newcommand{\Leq}[1]{\ensuremath{\underset{\textnormal{#1}}\leq}}
\newcommand{\Geq}[1]{\ensuremath{\underset{\textnormal{#1}}\geq}}
\newcommand{\Eq}[1]{\ensuremath{\underset{\textnormal{#1}}=}}

\newcommand{\tvd}[2]{\ensuremath{\norm{#1 - #2}_{tvd}}}
\newcommand{\Ot}{\ensuremath{\widetilde{O}}}
\newcommand{\eps}{\ensuremath{\varepsilon}}
\newcommand{\Paren}[1]{\Big(#1\Big)}
\newcommand{\Bracket}[1]{\Big[#1\Big]}
\newcommand{\bracket}[1]{\left[#1\right]}
\newcommand{\paren}[1]{\ensuremath{\left(#1\right)}\xspace}
\newcommand{\card}[1]{\left\vert{#1}\right\vert}
\newcommand{\Omgt}{\ensuremath{\widetilde{\Omega}}}

\newcommand{\IR}{\ensuremath{\mathbb{R}}}

\newcommand{\norm}[1]{\ensuremath{\|#1\|}}

\newcommand{\ceil}[1]{{\left\lceil{#1}\right\rceil}}

\newcommand{\set}[1]{\ensuremath{\left\{ #1 \right\}}}
\newcommand{\poly}{\mbox{\rm poly}}

\newcommand{\OPT}{\ensuremath{\textnormal{\mbox{\textsf{OPT}}}}\xspace}
\newcommand{\opt}{\textnormal{\ensuremath{\mbox{opt}}}\xspace}

\DeclareMathOperator*{\Exp}{\ensuremath{{\mathbb{E}}}}
\DeclareMathOperator*{\Prob}{\ensuremath{\textnormal{Pr}}}
\renewcommand{\Pr}{\Prob}

\newcommand{\Ex}{\Exp}

\newcommand{\etal}{et al.\xspace}

\newcommand{\dist}{\ensuremath{\mathcal{D}}}
\newcommand{\Yes}{\ensuremath{\textnormal{\textsf{Y}}}}
\newcommand{\No}{\ensuremath{\textnormal{\textsf{N}}}}
\newcommand{\distY}{\ensuremath{\mathcal{D}^{\Yes}}}
\newcommand{\distN}{\ensuremath{\mathcal{D}^{\No}}}

\newcommand{\FF}{\ensuremath{\mathcal{F}}}

\newcommand{\jstar}{\ensuremath{j^{\star}}}

\newcommand{\Prot}{\ensuremath{\Pi}}
\newcommand{\prot}{\ensuremath{\pi}}

\newcommand{\II}{\ensuremath{\mathbb{I}}}
\newcommand{\HH}{\ensuremath{\mathbb{H}}}

\newcommand{\Istar}{\ensuremath{I^{\star}}}

\newcommand{\event}{\ensuremath{\mathcal{E}}}

\renewcommand{\tvd}[2]{\ensuremath{\norm{#1-#2}}}

\newcommand{\errs}{\ensuremath{\textnormal{errs}}\xspace}

\newcommand{\DD}[2]{\ensuremath{\mathbb{D}(#1~||~#2)}}

\newcommand{\distribution}[1]{\ensuremath{\textnormal{dist}(#1)}}

\newcommand{\XX}{\ensuremath{\mathcal{X}}}

\renewcommand{\SS}{\ensuremath{\mathcal{S}}}

\newcommand{\CC}{\ensuremath{\mathcal{C}}}

\newcommand{\PP}{\ensuremath{\mathcal{P}}}

\newcommand{\cross}{\ensuremath{\textnormal{\textsf{X}}}}

\newcommand{\distp}{\dist^{\otimes}}

\renewcommand{\Yes}{\ensuremath{\textnormal{\textsf{Yes}}}\xspace}
\renewcommand{\No}{\ensuremath{\textnormal{\textsf{No}}}\xspace}

\newcommand{\rv}[1]{\ensuremath{{\mathsf{#1}}}}

\newcommand{\rProt}{\rv{\Prot}}

\newcommand{\rPhi}{\rv{\Phi}}

\newcommand{\rJ}{\rv{J}}

\newcommand{\rI}{\rv{I}}

\newcommand{\rIstar}{\rv{I}^{\star}}

\newcommand{\rRpub}{\ensuremath{\rv{R}^{\textnormal{\textsf{pub}}}}}
\newcommand{\rRpri}{\ensuremath{\rv{R}^{\textnormal{\textsf{pri}}}}}

\newcommand{\packing}{\ensuremath{\sigma}}
\newcommand{\labeling}{\ensuremath{\phi}}
\newcommand{\Labeling}{\ensuremath{\Phi}}

\newcommand{\Astar}{\ensuremath{A^{\star}}}

\newcommand{\rs}{Ruzsa-Szemer\'{e}di\xspace}

\title{Tight Bounds on the Round Complexity of the Distributed Maximum Coverage Problem}
\author{Sepehr Assadi\thanks{\ourinfo{\{sassadi,sanjeev\}@cis.upenn.edu}}  \and Sanjeev Khanna\footnotemark[1]}

\date{}

\begin{document}
\maketitle

\thispagestyle{empty}
\begin{abstract}
We study the maximum $k$-set coverage problem in the following distributed setting. A collection of sets $S_1,\ldots,S_m$ over a universe $[n]$ is partitioned
across $p$ machines and the goal is to find $k$ sets whose union covers the most number of elements. The computation proceeds in synchronous rounds. In each round, all machines simultaneously send a message to a central coordinator who then communicates back to all machines a summary to guide the computation for the next round. At the end of the last round, the coordinator outputs the answer. The main measures of efficiency in this setting are the approximation ratio of the returned solution, the communication cost of each machine, and the number of rounds of computation.

Our main result is an asymptotically tight bound on the tradeoff between these three measures for the distributed maximum coverage problem. We first show that any $r$-round protocol for this problem
either incurs a communication cost of $ k \cdot m^{\Omega(1/r)}$ or only achieves an approximation factor of $k^{\Omega(1/r)}$. This in particular implies that any protocol that simultaneously achieves good approximation ratio ($O(1)$ 
approximation) and good communication cost ($\Ot(n)$ communication per machine), essentially requires \emph{logarithmic} (in $k$) number of rounds. 
We complement our lower bound result by showing that there exist an $r$-round protocol that achieves an $\frac{e}{e-1}$-approximation (essentially best possible) with a communication cost of $k \cdot m^{O(1/r)}$ as well 
as an $r$-round protocol that achieves a $k^{O(1/r)}$-approximation with only $\Ot(n)$ communication per each machine (essentially best possible).

We further use our results in this distributed setting to obtain new bounds for the maximum coverage problem in two other main models of computation for massive
datasets, namely, the dynamic streaming model and the MapReduce model. 

\end{abstract}
\clearpage
\setcounter{page}{1}

\section{Introduction}\label{sec:intro}

A common paradigm for designing scalable algorithms for problems on massive data sets is to distribute the computation by partitioning the data across multiple machines interconnected via a communication network. The machines can then jointly compute a function on the union of their inputs by exchanging messages. A well-studied and important case of this paradigm is the \emph{coordinator model} (see, e.g.,~\cite{DolevF92,PhillipsVZ12,WoodruffZ13}). 
In this model, the computation proceeds in rounds, and in each round, all machines simultaneously send a message to a central coordinator who then communicates back to all machines a summary to guide the computation for the next round. At the end of the last round, the coordinator outputs the answer.  
Main measures of efficiency in this setting are the \emph{communication cost}, i.e., the total number of bits communicated by each machine, and the \emph{round complexity}, i.e., the number of rounds of 
computation. 

The distributed coordinator model (and the closely related \emph{message-passing} model\footnote{In absence of any restriction on round complexity, these two models are equivalent; see, e.g.,~\cite{PhillipsVZ12}.}) has been studied extensively in recent years (see, e.g.,~\cite{PhillipsVZ12,BravermanEOPV13,WoodruffZ12,WoodruffZ13,WoodruffZ14}, and references therein). Traditionally, the focus in this model has been on optimizing the communication cost and round complexity issues have been ignored. However, in recent years, motivated by application to big data analysis such as MapReduce computation, there have
been a growing interest in obtaining round efficient protocols for various problems in this model (see, e.g.,~\cite{AhnGM12Linear,AhnGM12,KapralovW14,IndykMMM14,GuchtWWZ15,MirrokniZ15,BarbosaENW15,AssadiKLY16,GuhaLZ17,AssadiK17}).  

In this paper, we study the \emph{maximum coverage} problem in the coordinator model: A collection of input sets $\SS:=\set{S_1,\ldots,S_m}$ over a
universe $[n]$ is arbitrarily partitioned across $p$ machines, and the goal is to select $k$ sets whose union covers the most number of elements from the universe. 
Maximum coverage is a fundamental optimization problem with a wide range of applications in various domains (see, e.g.,~\cite{KrauseG07,KempeKT03,SahaG09,EpastoLVZ17} for some applications). 
As an illustrative example of submodular maximization, the maximum coverage problem has been 
studied in various recent works focusing on scalable algorithms for massive data sets including in the coordinator model (e.g.,~\cite{IndykMMM14,MirrokniZ15}), MapReduce framework (e.g.,~\cite{ChierichettiKT10,KumarMVV13}), and the streaming model (e.g.~\cite{BateniEM17,McGregorV17}); see Section~\ref{sec:results} for a more comprehensive summary of previous results. 

Previous results for maximum coverage in the distributed model can be divided into two main categories: one on hand, we have \emph{communication efficient} protocols that only need $\Ot(n)$ communication and achieve a constant factor approximation, but require a \emph{large} number of rounds of $\Omega(p)$~\cite{BadanidiyuruMKK14,McGregorV17}\footnote{We remark that the algorithms of~\cite{BadanidiyuruMKK14,McGregorV17} are 
originally designed for the streaming setting and in that setting are quite efficient as they only require one or a constant number of passes over the stream. 
However, implementing one pass of a streaming algorithm in the coordinator model directly requires $p$ rounds of communication.}. On the other hand, 
we have \emph{round efficient} protocols that achieve a constant factor approximation in $O(1)$ rounds of communication, but incur a \emph{large} communication cost $k \cdot m^{\Omega(1)}$~\cite{KumarMVV13}.  

This state-of-the-affairs, namely, communication efficient protocols that require a large number of rounds, or round efficient protocols that require a large communication cost, raises the following
natural question: \emph{Does there exist a truly efficient distributed protocol for maximum coverage, that is, a protocol that simultaneously achieves $\Ot(n)$ communication cost, $O(1)$ round complexity, and gives a constant factor approximation?} This is the precisely the question addressed in this work. 

\subsection{Our Contributions}\label{sec:results}

Our first result is a \emph{negative} resolution of the aforementioned question. In particular, we show that,

\begin{result}\label{res:dist-lower}
For any integer $r \geq 1$, any $r$-round protocol for distributed maximum coverage either incurs $k \cdot m^{\Omega(1/r)}$ communication per machine or has an approximation factor of $k^{\Omega(1/r)}$. 
\end{result}

Prior to our work, the only known lower bound for distributed maximum coverage was due to McGregor and Vu~\cite{McGregorV17} who showed an $\Omega(m)$ communication lower bound for any protocol 
that achieves a better than $\paren{\frac{e}{e-1}}$-approximation (regardless of number of rounds and even if the input is \emph{randomly} distributed). 
Indyk~\etal~\cite{IndykMMM14} also showed that no composable coreset (a restricted family of single round protocols) can achieve 
a better than $\Omgt(\sqrt{k})$ approximation without communicating essentially the whole input (which is known to be tight~\cite{BarbosaENW15}). However, no super constant lower bounds on approximation ratio were known for this problem for arbitrary protocols even for one round of communication. Our result on the other hand implies that to achieve any constant factor approximation with any $O(n^{c})$ communication protocol (for a fixed constant $c > 0$), 
$\Omega\paren{\frac{\log{k}}{\log\log{k}}}$ rounds of communication are required.

In establishing Result~\ref{res:dist-lower}, we introduce a general framework for proving communication complexity lower bounds for \emph{bounded round} protocols in the distributed coordinator model. 
This framework, formally introduced in Section~\ref{sec:dgl}, captures many of the existing multi-party communication complexity lower bounds in the
literature for {bounded-round} protocols including~\cite{DobzinskiNO14,Konrad15,AssadiKLY16,AssadiKL17} (for one round a.k.a simultaneous protocols), and~\cite{AlonNRW15,Assadi17ca} (for multi-round protocols). We believe our framework will prove useful for establishing distributed lower bound results for other problems, and is thus interesting in its own right. 

We complement Result~\ref{res:dist-lower} by giving protocols that show that its bounds are essentially \emph{tight}. 

\begin{result}\label{res:dist-upper}
For any integer $r \geq 1$, there exist $r$-round protocols that achieve: 
\begin{enumerate}
	\item\label{part:dist-upper-one} an approximation factor of (almost) $\frac{e}{e-1}$ with $k \cdot m^{O({1/r})}$ communication per machine, or
	\item\label{part:dist-upper-two} an approximation factor of $O(r \cdot k^{1/r+1})$ with $\Ot(n)$ communication per machine. 
\end{enumerate}
\end{result}

Results~\ref{res:dist-lower} and~\ref{res:dist-upper} together provide a near complete understanding of the tradeoff between the approximation ratio, the communication cost, and the round complexity of 
protocols for the distributed maximum coverage problem for any fixed number of rounds. 

The first protocol in Result~\ref{res:dist-upper} is quite general in that it works for maximizing any monotone submodular function subject to a cardinality constraint. 
Previously, it was known how to achieve a $2$-approximation distributed algorithm for this problem with $m^{O(1/r)}$ communication and $r$ rounds of communication~\cite{KumarMVV13}. However, the previous best $\paren{\frac{e}{e-1}}$-approximation distributed algorithm for this problem with \emph{sublinear} in $m$ communication due to Kumar~\etal~\cite{KumarMVV13} requires at least $\Omega(\log{n})$ rounds of communication. As noted above, the $\paren{\frac{e}{e-1}}$ is {information theoretically} the best approximation ratio possible for any protocol that uses sublinear in $m$ communication~\cite{McGregorV17}. 

The second protocol in Result~\ref{res:dist-upper} is however tailored heavily to the maximum coverage problem. Previously, it was known that an $O(\sqrt{k})$ approximation can
be achieved via $\Ot(n)$ communication~\cite{BarbosaENW15} per machine, but no better bounds were known for this problem in multiple rounds under $\poly(n)$
communication cost. It is worth noting that since an adversary may assign all sets to a single machine, a communication cost of $\tilde{O}(n)$ is essentially best possible bound.
We now elaborate on some applications of our results. 

\paragraph{Dynamic Streams.} In the \emph{dynamic (set) streaming} model, at each step, either a new set is inserted or a previously inserted set is deleted from the stream. 
The goal is to solve the maximum coverage problem on the sets that are present at the end of the stream. A \emph{semi-streaming} algorithm is allowed to make one or a small number of passes over the stream and use 
only $O(n \cdot \poly\set{\log{m},\log{n}})$ space to process the stream and compute the answer. The streaming setting for the maximum coverage problem and the closely related set cover problem has been studied 
extensively in recent years
\cite{SahaG09,CormodeKW10,AusielloBGLP12,ChakrabartiK14,EmekR14,DemaineIMV14,BadanidiyuruMKK14,HarPeledIMV16,ChakrabartiW16,AssadiKL16,BateniEM17,ChenNZ16,McGregorV17,Assadi17sc,EpastoLVZ17}.
Previous work considered this problem in \emph{insertion-only} streams and more recently in the \emph{sliding window} model; to the best of our knowledge, no non-trivial results were
known for this problem in dynamic streams\footnote{A related problem of maximum $k$-vertex coverage, corresponding to picking $k$ vertices in a graph to cover the most number of 
edges, was very recently studied in~\cite{McGregorV17}. In this problem, the edges of the graph (corresponding to elements in maximum coverage) are being presented in a dynamic stream.}. 
Our Results~\ref{res:dist-lower} and~\ref{res:dist-upper} imply the first upper and lower bounds
for maximum coverage in dynamic streams.

Result~\ref{res:dist-lower} together with a recent characterization of multi-pass dynamic streaming algorithms~\cite{AiHLW16} proves that 
any semi-streaming algorithm for maximum coverage in dynamic streams that achieves any \emph{constant approximation} requires $\Omega\paren{\frac{\log{n}}{\log\log{n}}}$ passes over the stream.  
This is in sharp contrast with insertion-only streams in which semi-streaming algorithms can achieve (almost) $2$-approximation in only a single pass~\cite{BadanidiyuruMKK14} or (almost) $\paren{\frac{e}{e-1}}$-approximation
in a constant number of passes~\cite{McGregorV17} (constant factor approximations are also known in the sliding window model~\cite{ChenNZ16,EpastoLVZ17}). 
To our knowledge, this is the first \emph{multi-pass} dynamic streaming lower bound that is based on the characterization of~\cite{AiHLW16}. 
Moreover, as maximum coverage is a special case of submodular maximization (subject to cardinality constraint), our lower bound immediately extends to this problem and 
settles an open question of~\cite{EpastoLVZ17} on the space complexity of submodular maximization in dynamic streams.

We complement this result by showing that one can implement the first algorithm in Result~\ref{res:dist-upper} using proper \emph{linear sketches} in dynamic streams, which imply 
an (almost) $\paren{\frac{e}{e-1}}$-approximation semi-streaming algorithm for maximum coverage (and monotone submodular maximization) in $O(\log{m})$ passes. 
As a simple application of this result, we can
also obtain an $O(\log{n})$-approximation semi-streaming algorithm for the set cover problem in dynamic stream that requires $O(\log{m}\cdot\log{n})$ 
passes over the stream.

\paragraph{MapReduce Framework.} In the MapReduce model, there are $p$ machines each with a memory of size $s$ such that $p \cdot s = O(N)$, where $N$ is the total memory required to represent the input. 
MapReduce computation proceeds in synchronous rounds where in each round, each machine performs some local computation, and at the end of the round sends messages to other machine to guide the computation for the next round. The total size of messages received by each machine, however, is restricted to be $O(s)$. Following~\cite{KarloffSV10}, we require both $p$ and $s$ to at be at most $N^{1-\Omega(1)}$. 
The main complexity measure of interest in this model is typically the number of rounds. Maximum coverage and submodular maximization have also been extensively studied in the MapReduce model~\cite{ChierichettiKT10,BlellochPT11,KumarMVV13,MirzasoleimanKSK13,IndykMMM14,MirrokniZ15,BarbosaENW15,BarbosaENW16,BateniEM16a}. 

Proving round complexity lower bounds in the MapReduce framework turns out to be a challenging task (see, e.g.,~\cite{RoughgardenVW16} for implication of such lower bounds to long standing open problems
in complexity theory). As a result, most previous work on lower bounds concerns either communication cost (in a fixed number of rounds) or specific classes
of algorithms (for round lower bounds); see, e.g.,~\cite{AfratiSSU13,BeameKS13,PietracaprinaPRSU12,JacobLS14} (see~\cite{RoughgardenVW16} for more details). Our results
contribute to the latter line of work by characterizing the power of a large family of MapReduce algorithms for maximum coverage. 

Many existing techniques for MapReduce algorithms utilize the following paradigm which we call the \emph{sketch-and-update} approach: 
each machine sends a summary of its input, i.e., a sketch, to a \emph{single} designated machine which processes these sketches and computes a single combined sketch; 
the original machines then receive this combined sketch and update their sketch computation accordingly; this process is then continued on the updated sketches. Popular algorithmic techniques belonging
to this framework include composable coresets (e.g.,~\cite{BadanidiyuruMKK14,BalcanEL13,BateniBLM14,IndykMMM14}), the filtering method (e.g.,~\cite{LattanziMSV11}), linear-sketching
algorithms (e.g.,~\cite{AhnGM12Linear,AhnGM12,KapralovW14,AhnG15}), and the sample-and-prune technique (e.g.,~\cite{KumarMVV13,ImM15}), among many others. 
 
We use Result~\ref{res:dist-lower} to prove a lower bound on the power of this approach for solving maximum coverage in the MapReduce model. 
We show that any MapReduce algorithm for maximum coverage in the sketch-and-update framework that uses $s = m^{\delta}$ memory per machine
requires $\Omega(\frac{1}{\delta})$ rounds of computation. Moreover, both our algorithms in Result~\ref{res:dist-upper} belong to the sketch-and-update framework and can be implemented in the MapReduce model. 
In particular, the round complexity of our first algorithm for monotone submodular
maximization (subject to cardinality constraint) in Result~\ref{res:dist-upper} matches the best known algorithm of~\cite{BarbosaENW16} with the benefit of using
sublinear communication (the algorithm of~\cite{BarbosaENW16}, in each round, incurs a linear (in input size) communication cost). We remark that the algorithm in~\cite{BarbosaENW16} is however more general
in that it supports a larger family of constraints beside the cardinality constraint we study in this paper. 


\section{Preliminaries}\label{sec:prelim}

\paragraph{Notation.} For a collection of sets $\CC = \set{S_1,\ldots,S_t}$, we define $c(\CC) := \cup_{i \in [t]} S_i$, i.e., the set of elements covered by $\CC$. 
For a tuple $X= (X_1,\ldots,X_t)$ and index $i \in [t]$, $X^{<i} := (X_1,\ldots,X_{i-1})$ and $X^{-i} := (X_1,\ldots,X_{i-1},X_{i+1},\ldots,X_t)$. We use sans serif fonts to
denote random variables, i.e., $\rv{X}$. 

For a random variable $\rv{X}$ over a support $\Omega_{\rv{X}}$, $\distribution{\rv{X}}$ denotes the distribution of $\rv{X}$ and $\card{\rv{X}} := \log{\card{\Omega_{\rv{X}}}}$.
We use $\en{\rv{X}}$ and $\mi{\rv{X}}{\rv{Y}}$ to denote the \emph{Shannon entropy} of $\rv{X}$ and
\emph{mutual information} of $\rv{X}$ and $\rv{Y}$, respectively. For any two distributions $\mu$ and $\nu$ over the same probability space, 
$\DD{\mu}{\nu}$ and $\tvd{\mu}{\nu}$ denote the \emph{Kullback-Leibler divergence} and the \emph{total variation distance} between $\mu$ and $\nu$, respectively. 
 A summary of information theory facts used in this paper appears in Appendix~\ref{app:info}. 

\subsection{Communication Complexity Model}\label{app:communication} 

We prove our lower bound for distributed protocols using the framework of communication complexity, and in particular in the (number-in-hand) multiparty communication model
with shared blackboard: there are $p$ players (corresponding to machines) receiving inputs $(x_1,\ldots,x_p)$ from a prior distribution $\dist$ on $\XX_1 \times \ldots \XX_p$. 
The communication happens in rounds and in each round, the players \emph{simultaneously} write a message to a \emph{shared blackboard} visible to all parties. The message sent
by any player $i$ in each round can only depend on the input of the player, i.e., $x_i$, the current content of the blackboard, i.e., the messages communicated in previous rounds, and
public and private randomness. In addition to $p$ players, there exists a central party called the referee (corresponding to the coordinator) who only sees the content of the blackboard and public randomness and is 
responsible for outputting the answer in the final round. 

For a protocol $\prot$, we use $\Prot = (\Prot_1,\ldots,\Prot_p)$ to denote the transcript of the messages communicated by all players, i.e., the content of the blackboard. 
The \emph{communication cost} of a protocol $\prot$, denoted by $\norm{\prot}$, is the sum of worst-case length of the messages communicated by all players, i.e., $\norm{\prot} := \sum_{i=1}^{p} \card{\Prot_i}$. 
We further refer to $\max_{i \in [p]} \card{\Prot_i}$ as the \emph{per-player communication cost} of $\prot$. We remark that this model is identical to the distributed setting introduced earlier if we allow the coordinator to communicate with machines free of charge. As a result, communication lower bounds in this model
imply identical communication lower bounds for distributed protocols. We refer the reader to the excellent text by Kushilevitz and Nisan~\cite{KN97} for more details on communication complexity.

\subsection{Submodular Maximization with Cardinality Constraint}\label{app:submodular}

Let $V = \set{a_1,\ldots,a_m}$ be a ground set of $m$ items. For any set function $f: 2^V \rightarrow \IR$ and any $A \subseteq V$, we define 
the \emph{marginal contribution} to $f$ as a set function $f_A: 2^{V} \rightarrow \IR$ such that for all $B \subseteq V$, $f_A(B) = f(A \cup B) - f(A)$. 
When clear from the context, we abuse the notation and for $a \in V$, use $f(a)$ and $f_A(a)$ instead of $f(\set{a})$ and $f_A(\set{a})$, respectively. 
A function $f$ is \emph{submodular} iff for all $A \subseteq B \subseteq V$ and for all $a \in V$, $f_B(a) \leq f_A(a)$. 
A submodular function $f$ is additionally \emph{monotone} iff $\forall A \subseteq B \subseteq V$, $f(A) \leq f(B)$. 

The maximum coverage problem is a special case of {maximizing a monotone submodular function} subject to a cardinality
constraint of $k$, i.e., finding $\Astar \in \argmax_{A: \card{A} = k} f(A)$: for any set $S$ in maximum coverage we can have an item $a_S \in V$ and 
for each $A \subseteq V$, define $f(A) = \card{\bigcup_{a_S \in A} S}$. It is easy to verify that $f(\cdot)$ is monotone submodular. 

We use the following standard facts about monotone submodular functions in our proofs. 

\begin{fact}\label{fact:submodular-greedy}
	Let $f(\cdot)$ be a monotone submodular function, then: 
	\begin{align*}
		\forall A \subseteq V, B \subseteq V ~~ f(B) \leq f(A) + \sum_{a \in B \setminus A} f_A(a). 
	\end{align*}
\end{fact}

\begin{fact}\label{fact:submodular-subadditive}
	Let $f(\cdot)$ be a submodular function, then, for any $A \subseteq V$, $f_A(\cdot)$ is subadditive, i.e., $f_A(B \cup C) \leq f_A(B) + f_A(C)$ for all $B,C \subseteq V$. 
\end{fact}

\section{Technical Overview}\label{sec:techniques}

\paragraph{Lower Bounds (Result~\ref{res:dist-lower}).} Let us start by sketching our proof for simultaneous protocols. We provide each machine with a collection of sets
from a family of sets with small pairwise intersection such that \emph{locally}, i.e., from the perspective of each machine, all these sets look alike. At the same time, we ensure that \emph{globally}, one set in each machine
is \emph{special}; think of a special set as covering a \emph{unique} set of elements across the machines while all other sets are mostly covering a set of \emph{shared} elements. 
The proof now consists of two parts: $(i)$ use the simultaneity of the communication to argue that as each machine is oblivious to identity of its special set, it cannot convey enough information about
this set using limited communication, and $(ii)$ use the bound on the size of the intersection between the sets to show that this prevents the coordinator to find a good solution. 

The strategy outlined above is in fact at the core of many existing lower bounds for simultaneous protocols in the coordinator model including~\cite{DobzinskiNO14,Konrad15,AssadiKLY16,AssadiKL17} (a notable exception is the 
lower bound of~\cite{AssadiKL17} on estimating matching size in \emph{sparse} graphs). For example, to obtain the hard input distributions in~\cite{Konrad15,AssadiKLY16} for the maximum matching problem, we just need
to switch the sets in the small intersecting family above with induced matchings in a \rs graph~\cite{RuszaS78} (see also~\cite{AlonMS12} for more details on these graphs). 
The first part of the proof that lower bounds the communication cost required
for finding the special induced matchings (corresponding to special sets above), remains quite similar; however, we now need an entirely 
different argument for proving the second part, i.e., the bound obtained on the approximation ratio. This observation raises the following question: can we somehow ``automate'' the task of proving a communication lower 
bound in the arguments above so that one can focus solely on the second part of the argument, i.e., proving the approximation lower bound subject to each machine not being able to find its special entity, e.g., sets in the coverage 
problem and induced matchings in the maximum matching problem? 

We answer this question in the affirmative by designing a framework for proving communication lower bounds of the aforementioned type. 
We design an abstract hard input distribution using the ideas above and prove a general communication lower bound in this abstraction. This 
reduces the task of proving a communication lower bound for any specific problem to designing suitable combinatorial objects that roughly speaking enforce the importance of 
``special entities'' discussed above. We emphasize that this second part may still be a non-trivial challenge; for instance, lower bounds for matchings in~\cite{Konrad15,AssadiKLY16} rely on \rs graphs to prove
this part. Nevertheless, automating the task of proving a communication lower bound in our framework allows one to focus solely on a combinatorial problem
and entirely bypass the communication lower bounds argument. 

We further extend our framework to multi round protocols by building on the recent multi-party round elimination technique of~\cite{AlonNRW15} and its extension in~\cite{Assadi17ca}. 
At a high level, in the hard instances of $r$-round protocols, each machine is provided with a collection of instances of the same problem but on a  ``lower dimension'', i.e., defined on a smaller number of machines and input size. 
One of these instances is a special one in that it needs to be solved by the machines in order to solve the original instance. 
Again, using the simultaneity of the communication in one round, we show that the first round of communication cannot reveal enough information about this special instance and hence the machines
need to solve the special instance in only $r-1$ rounds of communication, which is proven to be hard inductively. Using the abstraction in our framework allows us to solely focus on the communication aspects of this argument,
independent of the specifics of the problem at hand. This allows us to provide a more direct and simpler proof than~\cite{AlonNRW15,Assadi17ca}, which is also applicable to a wider range of problems (the results
in~\cite{AlonNRW15,Assadi17ca} are for the setting of combinatorial auctions). However, although simpler than~\cite{AlonNRW15,Assadi17ca}, this proof is still
far from being simple - indeed, it requires a delicate information-theoretic argument (see Section~\ref{sec:dgl} for further details). This complexity of proving a multi-round lower bound in this model
is in fact another motivation for our framework. To our knowledge, the only previous lower bounds specific to bounded round protocols in the
coordinator model are those of~\cite{AlonNRW15,Assadi17ca}; we hope that our framework facilitates proving such lower bounds
in this model (understanding the power of bounded round protocols in this model is regarded as an interesting open question in the literature; see, e.g.,~\cite{WoodruffZ13}). 

Finally, we prove the lower bound for maximum coverage using this framework by designing a family of sets which we call \emph{randomly nearly disjoint}; roughly speaking the sets in this family have the property that
any suitably small \emph{random} subset of one set is essentially disjoint from any other set in the family. A reader familiar with~\cite{ChakrabartiW16} may realize that this definition is similar to the \emph{edifice} set-system
introduced in~\cite{ChakrabartiW16}; the main difference here is that we need every \emph{random} subsets of each set in the family to be disjoint from other sets, as opposed to a \emph{pre-specified} collection of
sets as in edifices~\cite{ChakrabartiW16}. As a result, the algebraic techniques of~\cite{ChakrabartiW16} do not seem suitable for our purpose and we prove our results using different techniques. The lower bound then follows by instantiating the hard distribution in our framework with this family for maximum coverage and proving the approximation lower bound.

\paragraph{Upper Bounds (Result~\ref{res:dist-upper}).} 
We achieve the first algorithm in Result~\ref{res:dist-upper}, namely an $\paren{\frac{e}{e-1}}$-approximation algorithm for maximum coverage (and submodular maximization), 
via an implementation of a thresholding greedy algorithm (see, e.g.,~\cite{BadanidiyuruV14,ChakrabartiW16}) in the distributed setting using the 
sample-and-prune technique of~\cite{KumarMVV13} (a similar thresholding greedy algorithm was used recently in~\cite{McGregorV17} for streaming maximum coverage). 
The main idea in the sample-and-prune technique is to sample a collection of sets from the machines in each round and send them to the coordinator who can build a partial greedy solution on those sets; the coordinator
then communicates this partial solution to each machine and in the next round the machines only sample from the sets that can have a substantial marginal contribution to the partial greedy solution maintained by the coordinator. 
Using a different greedy algorithm and a more careful choice of the threshold on the necessary marginal contribution from each set, we show that an $\paren{\frac{e}{e-1}}$-approximation can be obtained in constant number of rounds and sublinear communication (as opposed to the original approach of~\cite{KumarMVV13} which requires $\Omega(\log{n})$ rounds).

The second algorithm in Result~\ref{res:dist-upper}, namely a $k^{O(1/r)}$-approximation algorithm for any number of rounds $r$, 
however is more involved and is based on a new iterative sketching method specific to the maximum coverage problem. Recall that in 
our previous algorithm the machines are mainly ``observers'' and simply provide the coordinator with a sample of their input; our second algorithm is in some sense on the other extreme. 
In this algorithm, each machine is responsible for computing a suitable sketch of its input, which roughly speaking, is a collection of sets that tries to ``represent'' each optimal set in the input of this machine. 
The coordinator is also maintaining a greedy solution that is updated based on the sketches received from each machine. The elements covered by this collection are shared by the machines
to guide them towards the sets that are ``misrepresented'' by the sketches computed so far, and the machines update their sketches for the next round accordingly. We show that either the greedy solution
maintained by the coordinator is already a good approximation or the final sketches computed by the machines are now a good representative of the optimal sets and hence contain a good 
solution.

\section{A Framework for Proving Distributed Lower Bounds}\label{sec:dgl} 

We introduce a general framework
for proving communication complexity lower bounds for \emph{bounded round} protocols in the distributed coordinator model. 
Consider a \emph{decision} problem%
\footnote{While we present our framework for decision problems, with some modifications, it also extends to \emph{search} problems. We elaborate more on this in Appendix~\ref{app:discussions}.} $\PP$ defined by the family of functions $\PP_{s}: \set{0,1}^{s} \rightarrow \set{0,1}$ for any integer $s \geq 1$; 
we refer to $s$ as \emph{size} of the problem and to $\set{0,1}^s$ as its \emph{domain}. Note that $\PP_s$ can be a partial function, i.e., not necessarily defined on its whole domain. 
An instance $I$ of problem $\PP_s$ is simply a binary string of length $s$. We say that $I$ is a \Yes instance if $\PP_{s}(I) = 1$ and is a \No instance if $\PP_{s}(I) = 0$. 
For example, $\PP_s$ can denote the decision version of 
the maximum coverage problem over $m$ sets and $n$ elements with parameter $k$ (in which case $s$ would be a fixed function of $m$, $n$, and $k$ depending on the representation of the input) such
that there is a relatively large gap (as a function of, say, $k$) between the value of optimal solution in \Yes and \No instances. 
We can also consider the problem $\PP_s$ in the distributed model, whereby we distribute each instance between the players. The distributed coverage problem
for instance, can be modeled here by partitioning the sets in the instances of $\PP_s$ across the players.

To prove a communication lower bound for some problem $\PP$, one typically needs to design a hard input distribution
$\dist$ on instances of the problem $\PP$, and then show that distinguishing between the \Yes and \No cases in instances
sampled from $\dist$, with some sufficiently large probability, requires large communication.  Such a distribution inevitably depends on the specific
problem $\PP$ at hand.  We would like to \emph{abstract out} this dependence to the underlying problem and design a \emph{template} hard distribution for any
problem $\PP$ using this abstraction. Then, to achieve a lower bound for a particular problem $\PP$, one only needs to focus on the problem specific 
parts of this template and \emph{design} them according to the problem $\PP$ at hand. We emphasize that obviously we are not going to prove a 
communication lower bound for every possible distributed problem; rather, our framework reduces the problem of proving a communication lower bound for a problem $\PP$ 
to designing appropriate problem-specific \emph{gadgets} for $\PP$, which determine the strength of the lower bound one can ultimately prove using this framework. 
With this plan in mind, we now describe a high level overview of our framework. 
\subsection{A High Level Overview of the Framework}\label{sec:dgl-high-level}

Consider any decision problem $\PP$; we construct a recursive family of distributions $\dist_0,\dist_1,\ldots$ where $\dist_r$ is a hard input distribution for 
$r$-round protocols of $\PP_{s_r}$, i.e., for instances of size $s_r$ of the problem $\PP$, when the input is partitioned between $p_{r}$ players. Each instance in $\dist_r$ is a careful ``combination'' of many sub-instances of 
problem $\PP_{s_{r-1}}$ over different subsets of $p_{r-1}$ players, which are sampled (essentially) from $\dist_{r-1}$. 
We ensure that a small number of these sub-instances are ``special'' in that to solve the original instance of $\PP_{s_r}$, at least one of 
these instances of $\PP_{s_{r-1}}$ (over $p_{r-1}$ players) needs to be solved necessarily. We ``hide'' the special sub-instances in the input of players in a way that locally, no player is able to identify them 
and show that the first round of communication in any protocol with a small communication is spent only in identifying these special sub-instances. 
We then inductively show that as solving the special instance is hard for $(r-1)$-round protocols, the original instance must be hard for $r$-round protocols as well. 

We now describe this distribution in more detail. The $p_{r}$ players in the instances of distribution $\dist_r$ are partitioned into $g_r$ groups $P_1,\ldots,P_{g_r}$, each of
size $p_{r-1}$ (hence $g_r = p_r/p_{r-1}$). For every group $i \in [g_r]$ and every player $q \in P_i$, we create $w_r$ instances
$I^i_1,\ldots,I^i_{w_r}$ of the problem $\PP_{s_{r-1}}$ sampled from the distribution $\dist_{r-1}$. The domain of each instance $I^i_j$ is the same
across all players in $P_i$ and is different (i.e., disjoint) between any two $j \neq j' \in [w_r]$; we refer to $w_r$ as the \emph{width parameter}.  
The next step is to \emph{pack} all these instances into a single instance $I^i(q)$ for the player $q$; this is one of the places that we need 
a problem specific gadget, namely a \emph{packing function}\footnote{For a reader familiar with previous work in~\cite{AssadiKL17,AlonNRW15,Assadi17ca}, we note that a similar notion to a packing function is captured via 
a collection of disjoint blocks of vertices in~\cite{AlonNRW15} (for finding large matchings), \rs graphs in~\cite{AssadiKL17} (for estimating maximum matching size), and a family
of small-intersecting sets in~\cite{Assadi17ca} (for finding good allocations in combinatorial auctions).
In this work, we use the notion of randomly nearly disjoint set-systems defined in Section~\ref{sec:dist-lower}.} that can {pack} $w_r$ instances of problem $\PP_{s_{r-1}}$ into a single instance of problem $\PP_{s'_r}$
for some $s'_r \geq s_r$. 
We postpone the formal description of the packing functions to the next section, but roughly speaking, we require each player to be able to construct the instance $I^i(q)$ from the instances $I^i_1,\ldots,I^i_{w_r}$ and vice versa. 
As such, even though each player is given as input a single instance $I^i$, we can think of each player as conceptually ``playing'' in $w_r$ different instances $I^i_1,\ldots,I^i_{w_r}$ of $\PP_{s_{r-1}}$ instead. 

In each group $i \in [g_r]$, one of the instances, namely $I^i_{\jstar}$ for $\jstar \in [w_r]$, is the \emph{special instance} of the group: if we combine the inputs of players in $P_i$ on their
special instance $I^i_{\jstar}$, we obtain an instance which is sampled from the distribution $\dist_{r-1}$. On the other hand, all other instances are \emph{fooling instances}: if we combine the inputs of players in $P_i$
on their instance $I^i_{j}$ for $j \neq \jstar$, the resulting instance is \emph{not} sampled from $\dist_{r-1}$; rather, it is an instance created by picking the input of each player \emph{independently} from the corresponding marginal of $\dist_{r-1}$ (
$\dist_{r-1}$ is \emph{not} a product distribution, thus these two distributions are not identical). Nevertheless, by construction, each player is oblivious to this difference and hence is unaware of 
which instance in the input is the special instance (since the marginal distribution of a player's input is identical under the two distributions above). 

Finally, we need to combine the instances $I^1,\ldots,I^{g_r}$ to create the final instance $I$. To do this, we need another problem specific gadget, namely a \emph{relabeling function}. 
Roughly speaking, this function takes as input the index $\jstar$, i.e., the index of the special instances, and instances $I^1,\ldots,I^{g_r}$ and create the final instance $I$, while
``prioritizing'' the role of special instances in $I$. By prioritizing we mean that in this step, we need to ensure that the value of $\PP_{s_r}$ on $I$ is the same as the value of 
$\PP_{s_{r-1}}$ on the special instances. At the same time, we also need to ensure that this additional relabeling does not reveal the index of the special instance to each individual player, which requires a careful design
depending on the problem at hand. 

The above family of distributions is parameterized by the sequences $\set{s_r}$ (size of instances), $\set{p_r}$ (number of players), and $\set{w_r}$ (the width parameters), plus
the packing and relabeling functions. Our main result in this section is that if these sequences and functions satisfy some natural conditions (similar to what discussed above), then
any $r$-round protocol for the problem $\PP_{s_r}$ on the distribution $\dist_r$ requires $\Omega_r(w_r)$ communication. 

We remark that while we state our communication lower bound only in terms of $w_r$, to obtain any interesting lower bound using this technique, one needs to ensure 
that the width parameter $w_r$ is relatively large in the size of the instance $s_r$; this is also achieved by designing suitable packing and labeling functions (as well as a suitable representation of the problem). 
However, as ``relatively large'' depends heavily on the problem at hand, we do not add this requirement to the framework explicitly. 
A discussion on possible extensions of this framework as well as its connection to previous work appears in Appendix~\ref{app:discussions}.

\subsection{The Formal Description of the Framework}\label{sec:dgl-formal}

We now describe our framework formally. As stated earlier, to use this framework for proving a lower bound for any specific problem $\PP$, one needs to define appropriate
 problem-specific gadgets. These gadgets are functions that map multiple instances of $\PP_{s}$ to a single instance $\PP_{s'}$ for some $s' \geq s$. 
The exact application of these gadgets would become clear shortly in the description of our hard distribution for the problem $\PP$. 

\begin{definition}[Packing Function] \label{def:packing-gadget}
	For integers $s' \geq s \geq 1$ and $w \geq 1$, we refer to a function $\packing$ which maps 
	any tuple of instances $(I_1,\ldots,I_w)$ of $\PP_s$ to a single instance $I$ of $\PP_{s'}$ as a \emph{packing function} of \emph{width} $w$. 
\end{definition}

\begin{definition}[Labeling Family] \label{def:labeling-gadget}
	For integers $s'' \geq s' \geq 1$ and $g \geq 1$, we refer to a family of functions $\Labeling = \set{\labeling_i}$, where each $\labeling_i$ 
	is a function that maps any tuple of instances $(I^1,\ldots,I^g)$ of $\PP_{s'}$ to a single instance $I$ of $\PP_{s''}$ as a \emph{$g$-labeling family}, and to each function in this family, as a \emph{labeling function}. 
\end{definition}

We start by designing the following recursive family of hard distributions $\set{\dist_r}_{r\geq 0}$, parametrized
by sequences $\set{p_r}_{r\geq 0}$, $\set{s_r}_{r \geq 0}$, and $\set{w_r}_{r \geq 0}$. We require $\set{p_r}_{r \geq 0}$ and $\set{s_r}_{r \geq 0}$ to be \emph{increasing} sequences
and $\set{w_r}_{r \geq 0}$ to be \emph{non-increasing}. In two places marked in the distribution, we require one to design the aforementioned {problem-specific} gadgets for the distribution.

\textbox{Distribution $\dist_r$: \textnormal{A template hard distribution for $r$-round protocols of $\PP$ for any $r \geq 1$.}}{

\medskip


	\textbf{Parameters:} $p_r$: number of players, $s_r$: size of the instance, $w_r$: width parameter, $\packing_r$: packing function, and $\Labeling_r$: labeling family.

\begin{enumerate}
	\item Let $P$ be the set of $p_r$ players and define $g_r:= \frac{p_r}{p_{r-1}}$; partition the players in $P$ into $g_r$ \emph{groups} $P_1,\ldots,P_{g_r}$ each containing $p_{r-1}$ players. 
	\item \textbf{\emph{Design}} a packing function $\packing_r$ of width $w_{r}$ which maps $w_r$ instances of $\PP_{s_{r-1}}$ to an instance of $\PP_{s'_{r}}$ for some $s_{r-1} \leq s'_{r} \leq s_{r}$.  
			
	\item Pick an instance $\Istar_{r} \sim \dist_{r-1}$ over the set of players $[p_{r-1}]$ and domain of size $s_{r-1}$. 
	\item For each group $P_i$ for $i \in [g_r]$: 
		\begin{enumerate}
			\item Pick an index $\jstar \in [w_r]$ \emph{uniformly at random} and create $w_r$ instances $I^i_1,\ldots,I^i_{w_r}$ of problem $\PP_{s_{r-1}}$ as follows:
			\begin{enumerate}[(i)]
				\item Each instance $I^i_j$ for $j \in [w_r]$ is over the players $P_i$ and domain $D^i_j = \set{0,1}^{s_{r-1}}$. 
				\item For index $\jstar \in [w_r]$, $I^i_{\jstar} = \Istar_r$ by mapping (arbitrarily) $[p_{r-1}]$ to $P_i$ and domain of $\Istar_r$ to $D^i_{\jstar}$.  
				\item For any other index $j \neq \jstar$, $I^i_j \sim \distp_{r-1} := \otimes_{q \in P_i}\dist_{r-1}(q)$, i.e., the 
				\emph{product of marginal distribution} of the input to each player $q \in P_i$ in $\dist_{r-1}$. 
			\end{enumerate}
			\item Map all the instances $I^i_1,\ldots,I^i_{w_r}$ to a single instance $I^i$ using the function $\packing_r$. 
		\end{enumerate}
	\item \textbf{\emph{Design}} a \emph{$g_r$-labeling family} $\Labeling_r$ which maps $g_r$ instances of $\PP_{s'_r}$ to a \emph{single} instance $\PP_{s_{r}}$. 
	
	\item Pick a labeling function $\labeling$ from $\Labeling$ \emph{uniformly at random} and map the $g_r$ instances $I^1,\ldots,I^{g_r}$ of $\PP_{s'_r}$ to the output instance $I$ of $\PP_{s_{r}}$ using $\labeling$. 
	
	\item The input to each player $q \in P_i$ in the instance $I$, for any $i \in [g_r]$, is the input of player $q$ in the instance $I^i$, \emph{after} applying the mapping $\labeling$ to map $I^i$ to $I$. 
\end{enumerate}
}

We remark that in the above distribution, the ``variables'' in each instance sampled from $\dist_r$ are the instances $I^i_1,\ldots,I^i_{w_r}$ for all groups $i \in [g_r]$, the index
$\jstar \in [w]$, and both the choice of labeling family $\Labeling_r$ and the labeling function $\labeling$. On the other hand, the ``constants'' across all instances of $\dist_r$ are parameters $p_r,s_r$, and $w_r$, the
choice of grouping $P_1,\ldots,P_{g_r}$, and the packing function $\packing_r$. 

To complete the description of this recursive family of distributions, we need to explicitly define the distribution $\dist_0$ between $p_0$ players over $\set{0,1}^{s_0}$. 
We let $\dist_0 := \frac{1}{2} \cdot \distY_0 + \frac{1}{2} \cdot \distN_0$, where $\distY_0$ is a distribution over \Yes instances of $\PP_{s_0}$ and $\distN_0$ is a distribution over \No instances. 
The choice of distributions $\distY_0$ and $\distN_0$ are again \emph{problem-specific}. 

We start by describing the main properties of the packing and labeling functions that are required for our lower bound. 
For any player $q \in P_i$, define $I^i(q):= (I^i_1(q),\ldots,I^i_{w_r}(q))$, where for any $j \in [w_r]$, $I^{i}_j(q)$ denotes the input of player $q$ in the instance $I^i_j$. We require the packing and labeling functions to be
\emph{locally computable} defined as follows. 

\begin{definition}[Locally computable]\label{def:locally-computable}
	We say that the packing function $\packing_r$ and the labeling family $\Labeling_r$ are \emph{locally computable} iff any player $q \in P_i$ for $i \in [g_r]$, can compute
	the mapping of $I^i(q)$ to the final instance $I$, locally, i.e., only using $\packing_r$, the sampled labeling function $\labeling \in \Labeling_r$, and input $I^i(q)$. 
\end{definition}

We use $\labeling_q$ to denote the \emph{local mapping} of player $q \in P_i$ for mapping $I^i(q)$ to $I$; since $\packing_r$ is fixed in the distribution $\dist_r$, across different instances sampled
from $\dist_r$, $\labeling_q$ is only a function of $\labeling$. Notice that the input to each player $q \in P_i$ is \emph{uniquely} determined by $I^i(q)$ and $\labeling_q$.

Inside each instance $I$ sampled from $\dist_r$, there exists a unique \emph{embedded} instance $\Istar_{r}$ which is sampled from $\dist_{r-1}$. 
Moreover, this instance is essentially ``copied'' $g_r$ times, once in each instance $I^i_{\jstar}$ for each group $P_i$. We refer to the instance $\Istar_{r}$ as well as its copies $I^1_{\jstar},\ldots,I^{g_r}_{\jstar}$ 
as \emph{special instances} and to all other instances as \emph{fooling instances}. We require the packing and labeling functions to be \emph{preserving}, defined as,  

\begin{definition}[$\gamma$-Preserving]\label{def:preserving}
	We say that the packing function and the labeling family are \emph{$\gamma$-preserving} for a parameter $\gamma \in (0,1)$, iff 
	\begin{align*}
		\Pr_{I \sim \dist_{r}}\paren{\PP_{s_r}(I) = \PP_{s_{r-1}}(\Istar_r)} \geq 1-\gamma.
	\end{align*} 
	In other words, the value of $\PP_{s_r}$ on an instance $I$ should be equal to the value of $\PP_{s_{r-1}}$ on the embedded special instance $\Istar_r$ of $I$ w.p. $1-\gamma$.
\end{definition}

Recall that the packing function $\packing_r$ is a deterministic function that depends only on the distribution $\dist_r$ itself and not any specific instance (and hence the underlying special instances); on the other hand, 
the preserving property requires the packing and labeling functions to somehow ``prioritize'' the special instances over the fooling instances (in determining the value of the original instance). 
To achieve this property, the labeling family is allowed to vary based on the specific instance sampled from the distribution $\dist_r$. 
However, we need to limit the dependence of the labeling family to the underlying instance, which is captured through the definition of \emph{obliviousness} below. 

\begin{definition}\label{def:oblivious}
	We say that the labeling family $\Labeling_r$ is \emph{oblivious} iff it satisfies the following properties: 
	\begin{enumerate}[(i)]
		\item \label{item:def-obv-1} The only variable in $\dist_r$ which $\Labeling_r$ can depend on is $\jstar \in [w_r]$ (it can depend arbitrarily on the constants in $\dist_r$). 
		\item \label{item:def-obv-2} For any player $q \in P$, the local mapping $\labeling_q$ and $\jstar$ are \emph{independent} of each other in $\dist_r$. 
	\end{enumerate}
\end{definition}

Intuitively speaking, Condition~$(\ref{item:def-obv-1})$ above implies that a function $\labeling \in \Labeling_r$ can ``prioritize'' the special instances
based on the index $\jstar$, but it cannot use any further knowledge about the special or fooling instances. 
For example, one may be able to use $\labeling$ to distinguish special instances from other instances, i.e., determine $\jstar$, but would not be able to infer whether
the special instance is a \Yes instance or a \No one only based on $\labeling$. Condition~$(\ref{item:def-obv-2})$ on the other hand implies that for each player $q$, no information about the special
instance is revealed by the local mapping $\labeling_q$. This means that given the function $\labeling_q$ (and not $\labeling$ as a whole), one is not able to determine $\jstar$. 

Finally, we say that the family of distributions $\set{\dist_r}$ is a \emph{$\gamma$-hard recursive family}, iff $(i)$ it is parameterized by increasing sequences $\set{p_r}$ and $\set{s_r}$, and non-increasing sequence $\set{w_r}$, and 
$(ii)$, the packing and labeling functions in the family are locally computable, $\gamma$-preserving, and oblivious. 
We are now ready to present our main theorem of this section. 

\begin{theorem}\label{thm:dgl}
	Let $R \geq 1$ be an integer and suppose $\set{\dist_r}_{r=0}^{R}$ is a $\gamma$-hard recursive family for some $\gamma \in (0,1)$; 
	for any $r \leq R$, any $r$-round protocol for $\PP_{s_r}$ on $\dist_r$ which errs w.p. at most $1/3-r\cdot\gamma$ requires $\Omega(w_r/r^4)$ total communication. 
\end{theorem}

We prove Theorem~\ref{thm:dgl} in the next section. 

\subsection{Correctness of the Framework: Proof of Theorem~\ref{thm:dgl}}\label{sec:thm-dgl}

We first set up some notation. For any $r$-round protocol $\prot$ and any $\ell \in [r]$, we use $\rProt_\ell:= (\rProt_{\ell,1},\ldots,\rProt_{\ell,{p_r}})$ to denote the random variable for the transcript
of the message communicated by each player in round $\ell$ of $\prot$. We further use $\rPhi$ (resp. $\rPhi_q$) to denote the random variable for $\labeling$ (resp. local mapping $\labeling_q$) 
and $\rJ$ to denote the random variable for the index $\jstar$. Finally, for any $i \in [g_r]$ and $j \in [w_r]$, $\rI^{i}_j$ denotes the random variable for the instance $I^i_j$. 

We start by stating a simple property of oblivious mapping functions. 

\begin{proposition}\label{prop:dgl-oblivious}
	For any $i \in [g_r]$ and any player $q \in P_i$, conditioned on input $(I^i(q),\labeling_q)$ to player $q$, the index $\jstar \in [w_r]$ is chosen uniformly at random. 
\end{proposition} 
\begin{proof}
	By Condition~(\ref{item:def-obv-2}) of obliviousness in Definition~\ref{def:oblivious}, $\rPhi_q \perp \rJ$, and hence $\rJ \perp \rPhi_q = \labeling_q$. Moreover, by Condition~(\ref{item:def-obv-1}) of 
	Definition~\ref{def:oblivious}, $\rPhi_q$ cannot depend on $\rI^i(q)$ and hence $\rI^i(q) \perp \rPhi_q = \labeling_q$ also. Now notice that
	while the distribution of $\rI^i_{j}$ and $\rI^i_{\jstar}$ for $j \neq \jstar$, i.e., $\distp_{r-1}$ and $\dist_{r-1}$ are different, 
	the distribution of $\rI^i_j(q)$ and $\rI^i_{\jstar}(q)$ are identical by definition of $\distp_{r-1}$. As such, $\rI^{i}(q)$ and $\jstar$ are also independent of each other conditioned on $\rPhi_q = \labeling_q$, finalizing the proof. 
\end{proof}

We show that any protocol with a small communication cost cannot learn essentially any useful information about the special instance $\Istar_r$ in its first round. 

\begin{lemma}\label{lem:dgl-info-special}
	For any deterministic protocol $\prot$ for $\dist_r$, $\mi{\rIstar_r}{\rProt_1 \mid \rPhi,\rJ} \leq {\card{\rProt_1}}/{w_r}$.
\end{lemma}

\begin{proof}

The first step is to show that the information revealed about $\rIstar_r$ via $\rProt_1$ can be partitioned over the messages sent by each individual player
about their own input in their special instance. 

\begin{claim}\label{clm:dgl-separate}
	$\mi{\rIstar_r}{\rProt_1 \mid \rPhi,\rJ} \leq \sum_{q \in P}\mi{\rIstar_r(q)}{\rProt_{1,q} \mid \rPhi,\rJ}$. 
\end{claim}

\begin{proof}
	
Intuitively, the claim is true because after conditioning on $\rPhi$ and $\rJ$, the input of players become independent of each other on all fooling instances, i.e., every instance except for their copy of $\Istar_r$. 
As a result, the messages communicated by one player do not add extra information to messages of another one about $\Istar_r$. Moreover, since each player $q$ is observing $\Istar_r(q)$, the information revealed
by this player can only be about $\Istar_r(q)$ and not $\Istar_r$. We now provide the formal proof. 

	Recall that $\rProt_1 = (\rProt_{1,1},\ldots,\rProt_{1,p_r})$. By chain rule of mutual information, 
	\begin{align*}
		\mi{\rIstar_r}{\rProt_1 \mid \rPhi,\rJ} \Eq{\itfacts{chain-rule}} \sum_{q \in P} \mi{\rIstar_r}{\rProt_{1,q} \mid \rProt_1^{<q},\rPhi,\rJ}. 
	\end{align*}
	We first show that for each $q \in P$, 
	\begin{align}
	\mi{\rIstar_r}{\rProt_{1,q} \mid \rProt_1^{<q},\rPhi,J} \leq \mi{\rIstar_r}{\rProt_{1,q} \mid \rPhi,\rJ} \label{eq:clm-dgl-separate}.
	\end{align}
	
	Recall that, for any player $q$, $\rI(q)$ denotes the input to player $q$ in all instances in which $q$ is participating, and define $\rI(-q)$ as the collection of the inputs to all other players across all instances. 
	We argue that $\rI(q) \perp \rI({-q}) \mid \rIstar_r,\rPhi,\rJ$. The reason is simply because after conditioning on $\rIstar_r$, the only variables in $\rI(q)$ and $\rI({-q})$ are fooling instances that are 
	sampled from $\distp_{r-1}$ which is a product distribution across players. This implies that $\mi{\rI(q)}{\rI({-q}) \mid \rIstar_r,\rPhi,\rJ} = 0$ (by~\itfacts{info-zero}). Now, notice that the input to each player $q$ is 
	uniquely identified by $(\rI(q),\rPhi)$ (by locally computable property in Definition~\ref{def:locally-computable}) and hence
	conditioned on $\rIstar_r,\rPhi,J$, the message $\rProt_{1,q}$ is a deterministic function of $\rI(q)$. As such, by the data processing inequality (\itfacts{data-processing}), we have
	that $\mi{\rProt_{1,q}}{\rProt_1^{<q} \mid \rIstar_r,\rPhi,\rJ} = 0$; by Proposition~\ref{prop:info-decrease}, this implies Eq~(\ref{eq:clm-dgl-separate}) (here, conditioning on $\Prot_1^{<q}$ in RHS
	of Eq~(\ref{eq:clm-dgl-separate}) can only decrease the mutual information).
	
	Define $\rIstar_r(-q)$ as the input to all players in $\rIstar_r$ except for player $q$; hence $\rIstar_r = (\rIstar_r(q),\rIstar_r(-q))$. By chain rule of mutual information (\itfacts{chain-rule}), 
	\begin{align*}
		\mi{\rIstar_r}{\rProt_{1,q} \mid \rPhi,\rJ} = \mi{\rIstar_r(q)}{\rProt_{1,q} \mid \rPhi,\rJ} + \mi{\rIstar_r({-q})}{\rProt_{1,q} \mid \rIstar_r(q),\rPhi,\rJ} = \mi{\rIstar_r(q)}{\rProt_{1,q} \mid \rPhi,\rJ}
	\end{align*}
	since $\mi{\rIstar_r(-q)}{\rProt_{1,q} \mid \rIstar_r(q),\rPhi,\rJ} = 0$ as $\rProt_{1,q}$ is independent of $\rIstar_r(-q)$ after conditioning on $\rIstar_r(q)$ (and~\itfacts{info-zero}). The claim now follows from 
	Eq~(\ref{eq:clm-dgl-separate}) and above equation. 	
\end{proof}

Next, we use a direct-sum style argument to show that as each player is oblivious to the identity of the special instance in the input, 
the message sent by this player cannot reveal much information about the special instance, unless it is too large. 

\begin{claim}\label{clm:dgl-info-ds}
	For any group $P_i$ and player $q \in P_i$, $\mi{\rIstar_r(q)}{\rProt_{1,q} \mid \rPhi,\rJ} \leq \card{\rProt_{1,q}}/w_r$. 
\end{claim}
\begin{proof}
	We first argue that, 
	\begin{align}
		\mi{\rIstar_r(q)}{\rProt_{1,q} \mid \rPhi,\rJ} \leq \mi{\rIstar_r(q)}{\rProt_{1,q} \mid \rPhi_q,\rJ} \label{eq:lem-dgl-info-ds-1}.
	\end{align}
	
	Let $\rPhi = (\rPhi_q, \rPhi^{-q})$ where $\rPhi^{-q}$ denotes the rest of the mapping function $\rPhi$ beyond $\rPhi_q$. We have, $\rProt_{1,q} \perp \rPhi^{-q} \mid \rPhi_q,\rJ,\rIstar_r(q)$ 
	since after conditioning on $\rJ$, $\rPhi$ does not depend on any other variable in $\dist_{r}$ (by obliviousness property in Definition~\ref{def:oblivious}), 
	and hence the input to player $q$ and as a result $\rProt_{1,q}$ are independent of $\rPhi^{-q}$
	after conditioning on both $\rPhi_q$ and $\rJ$. Eq~(\ref{eq:lem-dgl-info-ds-1}) now follows from the independence of $\rProt_{1,q}$ and $\rPhi^{-q}$ and Proposition~\ref{prop:info-decrease} (as 
	conditioning on $\rPhi^{<q}$ in RHS of Eq~(\ref{eq:lem-dgl-info-ds-1}) can only decrease the mutual information). 
	
	We can bound the RHS of Eq~(\ref{eq:lem-dgl-info-ds-1}) as follows,  
	\begin{align*}
		\mi{\rIstar_r(q)}{\rProt_{1,q} \mid \rPhi_q,\rJ} &= \Ex_{j \in [w_r]}\Bracket{\mi{\rIstar_r(q)}{\rProt_{1,q} \mid \rPhi_q,\rJ = j}} = \frac{1}{w_r} \sum_{j=1}^{w_r} \mi{\rI^i_j(q)}{\rProt_{1,q} \mid \rPhi_q,\rJ = j} 
		\tag{$\jstar$ is chosen uniformly at random from $[w_r]$ and $\rIstar_r = \rI^i_j$ conditioned on $\rJ = j$}.
	\end{align*}
	
	Our goal now is to drop the conditioning on the event $\rJ = j$. By Definition~\ref{def:oblivious}, $\rPhi_q$ is independent of $\rJ = j$. Moreover, $\rI^i_j(q)$ is sampled from $\dist_{r-1}(q)$ (both in $\dist_{r-1}$ and 
	in $\distp_{r-1}$) and hence is independent of $\rJ=j$, even conditioned on $\rPhi_q$. Finally, by Proposition~\ref{prop:dgl-oblivious}, the input
	to player $q$ is independent of $\rJ=j$ and as $\rProt_{1,q}$ is a deterministic function
	of the input to player $q$, $\rProt_{1,q}$ is also independent of $\rJ=j$, even conditioned on $\rPhi_q$ and $\rI^i_j(q)$. This means that the joint distribution
	of $\rI^i_j(q),\rProt_{1,q}$, and $\rPhi_q$ is independent of the event $\rJ=j$ and hence we can drop this conditioning in the above term, and obtain that,
	\begin{align*}
			\frac{1}{w_r} \sum_{j=1}^{w_r} \mi{\rI^i_j(q)}{\rProt_{1,q} \mid \rPhi_i,\rJ = j}  &= \frac{1}{w_r} \sum_{j=1}^{w_r} \mi{\rI^i_j(q)}{\rProt_{1,q} \mid \rPhi_i} \\
			&\leq \frac{1}{w_r} \sum_{j=1}^{w_r} \mi{\rI^i_j(q)}{\rProt_{1,q} \mid \rI^{i,<j}(q), \rPhi_i} = \frac{1}{w_r} \cdot \mi{\rI^{i}(q)}{\rProt_{1,q} \mid \rPhi_i},
	\end{align*} 
	where the inequality holds since $\rI^i_j(q) \perp \rI^{i,<j}(q) \mid \rPhi_i$ and hence conditioning on $\rI^{i,<j}(q)$ can only increase the mutual information by Proposition~\ref{prop:info-increase}. 
	Finally, 
	\begin{align*}
		\frac{1}{w_r} \cdot \mi{\rI^{i}(q)}{\rProt_{1,q} \mid \rPhi_i} \Leq{\itfacts{uniform}} \frac{1}{w_1} \cdot \en{\rProt_{1,q} \mid \rPhi_i} \Leq{\itfacts{cond-reduce}} \frac{1}{w_r} \cdot \en{\rProt_{1,q}}
		 \Leq{\itfacts{uniform}} \frac{1}{w_r} \cdot \card{\rProt_{1,q}},
	\end{align*}
	finalizing the proof. 
\end{proof}

Lemma~\ref{lem:dgl-info-special} now follows from the previous two claims:

\begin{align*}
	\mi{\rIstar_r}{\rProt_1 \mid \rPhi,\rJ} \Leq{Claim~\ref{clm:dgl-separate}} \sum_{q \in P}\mi{\rIstar_r(q)}{\rProt_{1,q} \mid \rPhi,\rJ} \Leq{Claim~\ref{clm:dgl-info-ds}} 
	\frac{1}{w_r} \cdot \sum_{q \in P} \card{\rProt_{1,q}} = \frac{1}{w_r} \cdot \card{\rProt_1}. \qed
\end{align*}

\end{proof}

For any tuple $(\Prot_1,\labeling,j)$, we define the distribution $\psi(\Prot_1,\labeling,j)$ as the distribution of $\rIstar_r$ in $\dist_r$ conditioned on $\rProt_1 = \Prot_1$, $\rPhi = \labeling$, and $\rJ = j$. 
Recall that the original distribution of $\rIstar_r$ is $\dist_{r-1}$. In the following, we show that if the first message sent by the players is not too large, and hence does
not reveal much information by about $\Istar_r$ by Lemma~\ref{lem:dgl-info-special}, even after the aforementioned conditioning, distribution of $\rIstar_r$ does not change by much in average. 
Formally, 

\begin{lemma}\label{lem:dgl-istar-close}
	If $\card{\rProt_1} = o(w_r/r^4)$, then $\Ex_{(\Prot_1,\labeling,j)}\Bracket{\tvd{\psi(\Prot_1,\labeling,j)}{\dist_{r-1}}} = o(1/r^2)$. 
\end{lemma}
\begin{proof}
	Since $\Istar_r$ is independent of $\labeling$ and $\jstar$ in $\dist_r$, we have $\dist_{r-1} =  \distribution{\rIstar_r} = \distribution{\rIstar_r \mid \rPhi,\rJ}$. 
	As such, it suffices to show that $\distribution{\rIstar_r \mid \rPhi,\rJ}$ is close to the distribution of $\distribution{\rIstar_r \mid \rProt_1,\rPhi,\rJ}$. By Lemma~\ref{lem:dgl-info-special} and 
	the assumption $\card{\rProt_1} = o(w_r/r^4)$, we know that the information revealed about $\Istar_r$ by $\rProt_1$, conditioned on $\rPhi,\rJ$ is quite small, i.e., $o(1/r^4)$. This intuitively means that
	having an extra knowledge of $\rProt_1$ would not be able to change the distribution of $\Istar_r$ by much. We now formalizes this intuition. 
	\begin{align*}
		\Ex_{(\Prot_1,\labeling,j)}\Bracket{\tvd{\psi(\Prot_1,\labeling,j)}{\dist_{r-1}}} &= \Ex_{(\Prot_1,\labeling,j)}\Bracket{\tvd{\distribution{\rIstar_r \mid \Prot_1, \labeling, j}}{\distribution{\rIstar_r \mid \labeling,j}}} \\
		&\leq \Ex_{(\Prot_1,\labeling,j)}\Bracket{\sqrt{\frac{1}{2}\cdot\DD{\distribution{\rIstar_r \mid \Prot_1, \labeling, j}}{\distribution{\rIstar_r \mid \labeling,j}}}} \tag{By Pinsker's inequality (Fact~\ref{fact:pinskers})} \\
		&\leq \sqrt{\frac{1}{2} \cdot \Ex_{(\Prot_1,\labeling,j)}\Bracket{\DD{\distribution{\rIstar_r \mid \Prot_1, \labeling, j}}{\distribution{\rIstar_r \mid \labeling,j}}}} \tag{By concavity of $\sqrt{\cdot}$ and Jensen's inequality} \\
		&\Eq{Fact~\ref{fact:kl-info}} \sqrt{\frac{1}{2} \cdot \mi{\rIstar_r}{\rProt_1 \mid \rPhi,\rJ}} \Leq{Lemma~\ref{lem:dgl-info-special}} \sqrt{\frac{1}{2} \cdot \frac{1}{w_r} \cdot \card{\rProt_1}},
	\end{align*}
	which is $o(1/r^2)$ as $\card{\rProt_1} = o(w_r/r^4)$. 
\end{proof}

Define the  recursive function $\delta(r) := \delta(r-1) - o(1/r^2) - \gamma$ with base $\delta(0) = 1/2$. We have, 

\begin{lemma}\label{lem:dgl-round-elimination}
	For any deterministic $\delta(r)$-error $r$-round protocol $\prot$ for $\dist_r$, we have $\norm{\prot} = \Omega(w_r/r^4)$. 
\end{lemma}
\begin{proof}
	The proof is by induction on the number of rounds $r$. 
	
	\textbf{\emph{Base case:}} The base case of this lemma refers to $0$-round protocols for $\dist_0$, i.e., protocols that are not allowed any communication. As in the distribution $\dist_0$,
	\Yes and \No instances happen w.p. $1/2$ each and the coordinator has no input, any $0$-round protocol can only output the correct answer w.p. $1/2$, proving the induction base. 
	
	\textbf{\emph{Induction step:}} Suppose the lemma holds for all integers up to $r$ and we prove it for $r$ round protocols. The proof is by contradiction. 
	Given an $r$-round protocol $\prot_{r}$ violating the induction hypothesis, we create an $(r-1)$-round protocol $\prot_{r-1}$ which also violates the induction hypothesis, a contradiction. 
	Given an instance $I_{r-1}$ of $\PP_{s_{r-1}}$ over players $P^{r-1}$ and domain $D^{r-1} = \set{0,1}^{s_{r-1}}$, the protocol $\prot_{r-1}$ works as follows: 
	\begin{tbox}
	\begin{enumerate}
		\item Let $P^r = [p_r]$ and partition $P^r$ into $g_r$ equal-size groups $P_1,\ldots,P_{g_r}$ as is done in $\dist_{r}$. Create
		an instance $I_r$ of $\dist_{r}$ as follows: 
		\item Using \emph{public randomness}, the players in $P^{r-1}$ sample $R:= (\Prot_1,\labeling,\jstar) \sim (\distribution{\prot_{r}},\dist_r)$, i.e., from the (joint) distribution of protocol 
		$\prot_r$ over distribution $\dist_r$. 
		\item The $q$-th player in $P^{r-1}$ (in instance $I_{r-1}$) mimics the role of the $q$-th player in each group $P_i$ for $i \in [g_r]$ in $I_r$, denoted by player $(i,q)$, as follows: 
		\begin{enumerate}
			\item Set the input for $(i,q)$ in the special instance $I^{i}_{\jstar}(q)$ of $I_r$ as the original input of $q$ in $I_{r-1}$, i.e., $I_{r-1}(q)$ mapped via $\packing_r$ and $\labeling$ to
			 $I$ (as is done in $I_r$ to the domain $D^i_{\jstar}$). This is possible by the locally computable property of $\packing_r$ and $\labeling$ in Definition~\ref{def:locally-computable}. 
			\item \label{item:dgl-private-sampling} Sample the input for $(i,q)$ in all the fooling instances $I^{i}_j(q)$ of $I_r$ for any $j \neq \jstar$ using \emph{private randomness} from the \emph{correlated} distribution
			$\dist_{r} \mid \paren{\rIstar_{r} = I_{r-1}, (\rProt_1,\rPhi,\rJ) = R}$. This sampling is possible by Proposition~\ref{prop:dgl-sampling-possible} below. 
		\end{enumerate}
		\item Run the protocol $\prot_r$ from the second round onwards on $I_r$ assuming that in the first round the communicated message was $\Prot_1$ and output the same answer as $\prot_r$. 
	\end{enumerate}
	\end{tbox}

Notice that in Line~(\ref{item:dgl-private-sampling}), the distribution the players are sampling from depends on $\Prot_1,\phi,\jstar$ which are public knowledge (through sampling via public randomness), as well as
$\Istar_r$ which is \emph{not} a public information as each player $q$ only knows $\Istar_r(q)$ and not all of $\Istar_r$. Moreover, while random variables $\rI^i_j(q)$ (for $j \neq \jstar$) are originally independent across
different players $q$ (as they are sampled from the product distribution $\distp_{r-1}$), conditioning on the first message of the protocol, i.e., $\Prot_1$ correlates them, and hence a-priori it is not 
clear whether the sampling in Line~(\ref{item:dgl-private-sampling}) can be done without any further communication. Nevertheless, we can prove that this is the case and to sample from the distribution
in Line~(\ref{item:dgl-private-sampling}), each player only needs to know $\Istar_r(q)$ and not $\Istar_r$. 
\begin{proposition}\label{prop:dgl-sampling-possible}
	Suppose $\rI$ is the collection of all instances in the distribution $\dist_r$ and $\rI(q)$ is the input to player $q$ in instances in which $q$ participates; then,
	\begin{align*}
	\distribution{\rI \mid \rIstar_{r} = I_{r-1}, (\rProt_1,\rPhi,\rJ) = R} = \cross_{q \in P}~\distribution{\rI(q) \mid {\rIstar_r(q) = I_{r-1}(q),(\rProt_1,\rPhi,\rJ) = R}}. 
	\end{align*}	
\end{proposition}
\begin{proof}
	Fix any player $q \in P$, and recall that $\rI(-q)$ is the collection of the inputs to all players other than $q$ across all instances (special and fooling). We prove
	that $\rI(q) \perp \rI(-q) \mid (\rIstar_r(q), \rProt_1,\rPhi,\rJ)$ in $\dist_r$, which immediately implies the result. To prove this claim, by~\itfacts{info-zero}, it suffices to show 
	that $\mi{\rI(q)}{\rI(-q)\mid \rIstar_r(q), \rProt_1,\rPhi,\rJ} = 0$. Define $\rProt_1^{-q}$ as the set of all messages in $\rProt_1$ except for the message
	of player $q$, i.e., $\rProt_{1,q}$. We have, 
	\begin{align*}
		\mi{\rI(q)}{\rI(-q) \mid \rIstar_r(q), \rProt_1,\rPhi,\rJ} \leq \mi{\rI(q)}{\rI(-q) \mid \rIstar_r(q), \rProt_{1,q},\rPhi,\rJ}, 
	\end{align*}
	since $\rI(q) \perp \rProt_1^{-q} \mid \rI(-q),\rIstar_r(q), \rProt_{1,q},\rPhi,\rJ$ as the input to players $P \setminus \set{q}$ is uniquely
	determined by $\rI(-q),\rPhi$ (by the locally computable property in Definition~\ref{def:locally-computable}) and hence $\rProt_1^{-q}$ is deterministic after the conditioning; this independence 
	means that conditioning on $\rProt_1^{-q}$ in the RHS above can only decrease the mutual information by Proposition~\ref{prop:info-decrease}.  
	We can further bound the RHS above by,
	\begin{align*}
		\mi{\rI(q)}{\rI(-q)\mid \rIstar_r(q), \rProt_{1,q},\rPhi,\rJ} \leq \mi{\rI(q)}{\rI(-q)\mid \rIstar_r(q),\rPhi,\rJ},
	\end{align*}
	since ${\rI(-q) \perp \rProt_{1,q} \mid \rI(q),\rIstar_r(q), \rPhi,\rJ}$ as the input to player $q$ is uniquely determined by $\rI(q),\rPhi$ (again by Definition~\ref{def:locally-computable}) and 
	hence after the conditioning, $\rProt_{1,q}$ is deterministic; this implies that conditioning on $\rProt_{1,q}$ in RHS above can only decrease the mutual information by Proposition~\ref{prop:info-decrease}. 
	Finally, observe that $\mi{\rI(q)}{\rI(-q)\mid \rIstar_r(q),\rPhi,\rJ} = 0$ by \itfacts{info-zero}, since after conditioning on $\Istar_r(q)$, 
	the only remaining instances in $\rI(q)$ are fooling instances which are sampled from the distribution $\distp_{r-1}$ which is independent across the players. This
	implies that $\mi{\rI(q)}{\rI(-q) \mid \rIstar_r(q), \rProt_1,\rPhi,\rJ} = 0$ also which finalizes the proof. 
\end{proof}

Having proved Proposition~\ref{prop:dgl-sampling-possible}, it is now easy to see that $\prot_{r-1}$ is indeed a valid $r-1$ round protocol for
distribution $\dist_{r-1}$: each player $q$ can perform the sampling in Line~(\ref{item:dgl-private-sampling}) without any communication as $(\Istar(q),\Prot_1,\Phi,J)$ are all known to $q$; this allows the players to simulate the first 
round of protocol $\prot_r$ without any communication and hence only need $r-1$ rounds of communication to compute the answer of $\prot_r$. We can now prove that, 

\begin{claim}\label{clm:dgl-embedding}
	Assuming $\prot_r$ is a $\delta$-error protocol for $\dist_r$, $\prot_{r-1}$ would be a $\paren{\delta+\gamma + o(1/r^2)}$-error protocol for $\dist_{r-1}$. 
\end{claim}

\begin{proof}	
	Our goal is to calculate the probability that $\prot_{r-1}$ errs on an instance $I_{r-1} \sim \dist_{r-1}$. For the sake of analysis, suppose that $I_{r-1}$ is 
	instead sampled from the distribution $\psi$ for a randomly chosen tuple $(\Prot_1,\labeling,\jstar)$ (defined before Lemma~\ref{lem:dgl-istar-close}). Notice that by Lemma~\ref{lem:dgl-istar-close}, these two distributions
	are quite close to each other in total variation distance, and hence if $\prot_{r-1}$ has a small error on distribution $\psi$ it would necessarily has a small error on $\dist_{r-1}$ as well (by Fact~\ref{fact:tvd-small}). 
	
	Using Proposition~\ref{prop:dgl-sampling-possible}, it is easy to verify that if $I_{r-1}$ is sampled from $\psi$, then the
	instance $I_r$ constructed by $\prot_{r-1}$ is sampled from $\dist_{r}$ and moreover $\Istar_r = I_{r-1}$. As such, since $(i)$ $\prot_r$ is a 
	$\delta$-error protocol for $\dist_{r}$, $(ii)$ the answer to $I_r$ and $\Istar_r=I_{r-1}$ are the same w.p. $1-\gamma$ (by $\gamma$-preserving property in Definition~\ref{def:preserving}), 
	and $(iii)$ $\prot_{r-1}$ outputs the same answer as $\prot_r$, protocol $\prot_{r-1}$ is 
	a $(\delta+\gamma)$-error protocol for $\psi$. 
	
	We now prove this claim formally. Define $\rRpri$ and $\rRpub$ as, respectively, the private and public randomness used by $\prot_{r-1}$. We have,  
	\begin{align*}
		\Pr_{\dist_{r-1}}\paren{\prot_{r-1}~\errs} &= \Ex_{\rI_{r-1} \sim \dist_{r-1}} \Ex_{\rRpub}\Bracket{\Pr_{\rRpri}\paren{\prot_{r-1}~\errs \mid \rRpub}} \\
		&= \Ex_{(\Prot_1,\labeling,\jstar)} \Ex_{\rI_{r-1} \sim \dist_{r-1} \mid (\Prot_1,\labeling,\jstar)} \Bracket{\Pr_{\rRpri}\paren{\prot_{r-1}~\errs \mid \Prot_1,\labeling,\jstar}} 
		\tag{as $\rRpub \perp \rI_{r-1}$ and $\rRpub = (\Prot_1,\psi,\jstar)$ in protocol $\prot_{r-1}$} \\
		 &\leq \Ex_{(\Prot_1,\labeling,\jstar)} \Bracket{ \Ex_{\rI_{r-1} \sim \psi(\Prot_1,\labeling,\jstar)} \Bracket{\Pr_{\rRpri}\paren{\prot_{r-1}~\errs \mid \Prot_1,\labeling,\jstar}} + \tvd{\dist_{r-1}}{\psi(\Prot_1,\labeling,\jstar)}} 
		 \tag{by Fact~\ref{fact:tvd-small} for distributions $\dist_{r-1}$ and $\psi(\Prot_1,\labeling,\jstar)$} \\
		 &= \Ex_{(\Prot_1,\labeling,\jstar)} \Ex_{\rI_{r-1} \sim \psi(\Prot_1,\labeling,\jstar)} \Bracket{\Pr_{\rRpri}\paren{\prot_{r-1}~\errs \mid \Prot_1,\labeling,\jstar}} + o(1/r^2) 
		 \tag{by linearity of expectation and Lemma~\ref{lem:dgl-istar-close}} \\
		 &= \Ex_{(\Prot_1,\labeling,\jstar)} \Ex_{\rI_{r-1} \sim \psi(\Prot_1,\labeling,\jstar)} \Bracket{\Pr_{\dist_r}\paren{\prot_{r-1}~\errs \mid \rIstar_r = I_{r-1},\Prot_1,\labeling,\jstar}} + o(1/r^2) 
		 \tag{$\distribution{\rRpri} = \dist_r \mid \rIstar_r = I_{r-1},\Prot_1,\labeling,\jstar$} \\
		 &\leq \Ex_{(\Prot_1,\labeling,\jstar)} \Ex_{\rI_{r-1} \sim \psi(\Prot_1,\labeling,\jstar)} \Bracket{\Pr_{\dist_r}\paren{\prot_{r}~\errs \mid \rIstar_r = I_{r-1},\Prot_1,\labeling,\jstar}} + \gamma + o(1/r^2) 
		 \tag{$\PP_{s_r}(I_r) = \PP_{s_{r-1}}(I_{r-1})$ w.p. $1-\gamma$ by Definition~\ref{def:preserving} and $\prot_{r-1}$ outputs the same answer as $\prot_r$} \\
		 &= \Ex_{(\Istar_r,\Prot_1,\labeling,\jstar) \sim \dist_r} \Bracket{\Pr_{\dist_r}\paren{\prot_{r}~\errs \mid \rIstar_r = I_{r-1},\Prot_1,\labeling,\jstar}} + \gamma + o(1/r^2) 
		 \tag{$\psi(\Prot_1,\labeling,\jstar) = \distribution{\rIstar_r \mid \Prot_1,\labeling,\jstar}$ in $\dist_r$ by definition} \\
		 &= \Pr_{\dist_r}\paren{\prot_{r}~\errs} + o(1/r^2) \leq \delta + \gamma + o(1/r^2), \tag{as $\prot_r$ is a $\delta_r$-error protocol for $\dist_r$ by the assumption in the lemma statement}
	\end{align*}
	finalizing the proof. 
\end{proof}

We are now ready to finalize the proof of Lemma~\ref{lem:dgl-round-elimination}. Suppose $\prot_r$ is a deterministic $\delta(r)$-error protocol for $\dist_{r}$ with communication cost $\norm{\prot_r} = o(w_r/r^4)$. 
By Claim~\ref{clm:dgl-embedding}, $\prot_{r-1}$ would be a randomized $\delta(r-1)$-error protocol for $\dist_{r-1}$ with $\norm{\prot_{r-1}} \leq \norm{\prot_r}$ (as $\delta(r-1) = \delta(r) + \gamma + o(1/r^2)$). 
By an averaging argument, we can fix the randomness in $\prot_{r-1}$ to obtain a deterministic protocol $\prot'_{r-1}$ over the distribution $\dist_{r-1}$ with the same error $\delta(r-1)$ and communication 
of $\norm{\prot'_{r-1}} = o(w_r/r^4) = o(w_{r-1}/r^4)$ (as $\set{w_r}_{r \geq 0}$ is a non-increasing sequence). But such a protocol contradicts the induction hypothesis for $(r-1)$-round protocols, finalizing the proof. 
\end{proof}

\begin{proof}[Proof of Theorem~\ref{thm:dgl}]
	By Lemma~\ref{lem:dgl-round-elimination}, any deterministic $\delta(r)$-error $r$-round protocol for $\dist_r$ requires $\Omega(w_r/r^4)$ total communication. This immediately extends
	to randomized protocols by an averaging argument, i.e., the easy direction of Yao's minimax principle~\cite{Yao79}. The statement in the theorem now follows from this since for any $r \geq 0$, 
	$\delta(r) = \delta(r-1)-\gamma-o(1/r^2) = \delta(0)-r\cdot\gamma-\sum_{\ell=1}^{r} o(1/\ell^2)= 1/2 - r\cdot \gamma - o(1) > 1/3 - r \cdot \gamma$ (as $\delta(0) = 1/2$ and $\sum_{\ell=1}^{r}1/\ell^2$ is a converging series and hence is bounded by some absolute constant independent of $r$). 
\end{proof}

\newcommand{\MM}{\ensuremath{\mathcal{M}}}

\renewcommand{\AA}{\ensuremath{\mathcal{A}}}

\newcommand{\Ustar}{\ensuremath{U^{\star}}}

\newcommand{\coverage}[2]{\ensuremath{\textnormal{\textsf{coverage}}(#1,#2)\xspace}}

\section{A Distributed Lower Bound for Maximum Coverage}\label{sec:lb-dist}

We prove our main lower bound for maximum coverage in this section, formalizing Result~\ref{res:dist-lower}.

\begin{theorem}\label{thm:dist-lower}
	For integers $1 \leq r , c \leq o\paren{\frac{\log{k}}{\log\log{k}}}$ with $c \geq 4r$, any $r$-round protocol for the maximum coverage problem that can approximate the value of optimal solution
	to within a factor of better than $\paren{\frac{1}{2c} \cdot \frac{k^{1/2r}}{\log{k}}}$ w.p. at least $3/4$ requires $\Omega\paren{\frac{k}{r^4} \cdot m^{\frac{c}{(c+2)\cdot 4r}}}$ communication per machine.
	The lower bound applies to instances with $m$ sets, $n = m^{1/\Theta(c)}$ elements, and $k = \Theta(n^{2r/(2r+1)})$. 
\end{theorem}

The proof is based on an application of Theorem~\ref{thm:dgl}. In the following, let $c \geq 1$ be any integer (as in Theorem~\ref{thm:dist-lower}) 
and $N \geq 12c^2$ be a sufficiently large integer which we use to define the main parameters for our problem. To invoke Theorem~\ref{thm:dgl}, we need to instantiate the
recursive family of distributions $\set{\dist_r}_{r=0}^{c}$ in Section~\ref{sec:dgl} with appropriate sequences and gadgets for the maximum coverage problem.  
We first define sequences (for all $0 \leq r \leq c$): 
\begin{tbox}
\vspace{-20pt}
\begin{align*}
	k_r = p_r = (N^2-N)^{r}, ~~~ n_r = N^{2r+1}, ~~~ m_r = \paren{N^{c} \cdot (N^2-N)}^{r}, ~~~ w_r = N^c ~~~ g_r = (N^2 - N)
\end{align*}
\end{tbox}
Here, $m_r$, $n_r$, and $k_r$, respectively represent the number of sets and elements and the parameter $k$ in the maximum coverage problem in the instances of each distribution $\dist_r$ and together can identify the size of each instance (i.e., the parameter ${s_r}$ defined in Section~\ref{sec:dgl} for the distribution ${\dist_r}$). Moreover, ${p_r},w_r$ and $g_r$ represent
the number of players, the width parameter, and the number of groups in $\dist_r$, respectively (notice that $g_r = p_r/p_{r-1}$ as needed in distribution $\dist_r$).  

Using the sequences above, we define: 
\begin{tbox}
	$\coverage{N}{r}$: the problem of deciding whether the optimal $k_r$ cover of universe $[n_r]$ with 
	$m_r$ input sets is at least $\paren{k_r \cdot N}$ (\Yes case), or at most $\paren{k_r \cdot 2c \cdot \log{(N^{2r})}}$ (\No case).
\end{tbox}

Notice that there is a gap of roughly $N \approx k_r^{1/2r}$ (ignoring the lower order terms) between the value of the optimal solution in \Yes and \No cases of
$\coverage{N}{r}$. We prove a lower bound for deciding between \Yes and \No instances of $\coverage{N}{r}$, when the {input sets are partitioned} between the players, 
which implies an identical lower bound for algorithms that can approximate the value of optimal solution in maximum coverage to within a factor smaller than (roughly) $k_r^{1/2r}$. 

Recall that to use the framework introduced in Section~\ref{sec:dgl}, one needs to define two problem-specific gadgets, i.e., a packing function, and a labeling family. In the following section, 
we design a crucial building block for our packing function. 

\paragraph{RND Set-Systems.} Our packing function is based on the following set-system. 

\begin{definition}\label{def:rnd}
	For integers $N,r,c \geq 1$, an $(N,r,c)$-\emph{randomly nearly disjoint (RND)} set-system over a universe $\XX$ of $N^{2r}$ elements, is a collection $\SS$ of subsets of $\XX$ satisfying the following properties: 

\begin{enumerate}[(i)]
	\item Each set $A \in \SS$ is of size $N^{2r-1}$. 
	\item \label{item:rnd-2} Fix any set $B \in \SS$ and suppose $\CC_{B}$ is a collection of $N^{c \cdot r}$ subsets of $\XX$ whereby each set in $\CC_{B}$ is chosen 
	by picking an arbitrary set $A \neq B$ in $\SS$, and then picking an $N$-subset uniformly at random from $A$ (we do \emph{not} assume independence between the sets in $\CC_B$). Then, 
	\begin{align*}
		\Pr\Paren{\exists~S \in \CC_B \text{~s.t.~} \card{S \cap B} \geq 2  c \cdot r \cdot \log{N}} = o(1/N^3).
	\end{align*}
	Intuitively, this means that any random $N$-subset of some set $A \in \SS$ is essentially disjoint from any other set $B \in \SS$ w.h.p. 
\end{enumerate}
\end{definition}

We prove an existence of large RND set-systems. 

\begin{lemma}\label{lem:rnd-size}
	For integers $1 \leq r \leq c$ and sufficiently large integer $N \geq c$, there exists an $(N,r,c)$-RND set-system $\SS$ of size $N^{c}$ over any universe $\XX$ of size $N^{2r}$.  
\end{lemma}
\begin{proof}
We use a probabilistic argument to prove this lemma. 
	First, construct a collection $\SS'$ of $N^{c}$ subsets of $\XX$, each chosen independently and uniformly at random from all $(N^{2r-1})$-subsets
	of $\XX$. The proof is slightly different for the case when $r=1$ and for larger values of $r > 1$. In the following, we prove the result for the
	more involved case of $r > 1$ and then sketch the proof for the $r=1$ case. 
	
	We start with the following simple claim. 
	
	\begin{claim}\label{clm:rnd-clm1}
		For any two sets $A,B \in \SS'$, 
		\begin{align*}
			\Pr\paren{\card{A \cap B} \geq 2N^{2r-2}} \leq \exp\paren{-{2N^{2r-2}}}
		\end{align*}
	\end{claim}
	\begin{proof}
		Fix a set $A \in \SS'$ and pick $B$ uniformly at random from all $(N^{2r-1})$-subsets of $\XX$ (as is the construction in $\SS'$ since $A$ and $B$ are chosen independently). 
		For any element $i \in A$, we define an indicator random variable $X_i \in \set{0,1}$ which is $1$ iff $i \in B$ as well. Moreover, we define $X:= \sum_{i \in  A} X_i$ to denote size of $\card{A \cap B}$. 
	
		By the choice of $B$, we have $\Ex\bracket{X} = \sum_{i \in A} \Ex\bracket{X_i} = \sum_{i \in A} \frac{1}{N} = N^{2r-2}$. Moreover, it is straightforward to verify that 
		the random variables $X_i$ are \emph{negatively correlated}; as such, we can apply Chernoff bound to obtain that, 
		\begin{align*}
			\Pr\paren{\card{A \cap B} \geq 2N^{2r-2}} = \Pr\paren{X \geq 2\Ex\bracket{X}} \leq \exp\paren{-2\Ex\bracket{X}} = \exp\paren{-2N^{2r-2}}
		\end{align*}
		finalizing the proof. 
	\end{proof}
	
	By Claim~\ref{clm:rnd-clm1} and taking a union bound over all ${{N^{c}}\choose{2}}$ pairs of subsets $A,B \in \SS'$, the probability that there exists two subsets $A,B \in \SS'$ with $\card{A \cap B} \geq 2N^{r-2}$ is at most, 
	\begin{align*}
		{{N^{c}}\choose{2}} \cdot \exp\paren{-2N^{2r-2}} \leq \exp\paren{-2N^{2r-2}+2c \cdot \log{N}} < 1
	\end{align*}
	as $r \geq 2$ and $c \leq N$. This in particular implies that there exists a collection $\SS$ of $N^{c}$ many $(N^{2r-1}$)-subsets of $\XX$ such that for any two sets $A,B \in \SS$, $\card{A \cap B} \leq 2N^{2r-2}$. 
	We fix this $\SS$ as our target collection and prove that it satisfies Property~(\ref{item:rnd-2}) of Definition~\ref{def:rnd} as well. 
	
	Fix any $B \in \SS$ and define $\CC_B$ as in Definition~\ref{def:rnd}. We prove that, 
	\begin{claim}\label{clm:rnd-clm2}
		For any set $S \in \CC_B$, 
		\begin{align*}
			\Pr\paren{\card{S \cap B} \geq 2c \cdot r \cdot \log{N}} \leq \exp\paren{-c \cdot r \cdot \log{N}} 
		\end{align*}
	\end{claim}
	\begin{proof}
		The proof is similar to Claim~\ref{clm:rnd-clm1}. Suppose $S$ is chosen from some arbitrary set $A \in \SS \setminus \set{B}$.
		Note that $S \cap B \subseteq A \cap B$. 
		For any element $i \in A \cap B$, define a random variable $X_i \in \set{0,1}$ which is $1$ iff $i \in S$ as well. Define $X:= \sum_{i \in A \cap B}X_i$ which denotes
		the size of $S \cap B$. We have, 
		\[
		\Ex\bracket{X} = \sum_{i \in A \cap B} \frac{\card{S}}{\card{A}} = \card{A \cap B} \cdot \frac{N}{N^{2r-1}} \leq 2
		\]
		as $\card{A \cap B} \leq 2N^{2r-2}$ by the property of the collection $\SS$. Again, using the fact that $X_i$ variables are negatively correlated, we can apply Chernoff bound 
		and obtain that, 
		\begin{align*}
			\Pr\paren{X \geq 2c \cdot r \cdot \log{N}} = \Pr\paren{X \geq c \cdot r\cdot \log{N} \cdot \Ex\bracket{X}} \leq \exp\paren{-c\cdot r \cdot \log{N}}
		\end{align*}
		finalizing the proof. 
	\end{proof}
	
	To obtain the final result for $r > 1$ case, we can use Claim~\ref{clm:rnd-clm2} and take a union bound on the $N^{c \cdot r}$ possible choices for the 
	set $S$ in $\CC_B$ and obtain that, 
	\begin{align*}
			\Pr\Paren{\exists~S \in \CC_B \text{~s.t.~} \card{S \cap B} \geq 2 \cdot c \cdot r \cdot \log{N}} \leq N^{c \cdot r} \cdot \exp\paren{-c \cdot r \cdot \log{N}} = o(1/N^3)
	\end{align*}
	for sufficiently large $N$. 
	
	To obtain the result when $r = 1$, we can show, exactly as in Claim~\ref{clm:rnd-clm1}, that for any two sets $A,B \in \SS'$, 
		\begin{align*}
			\Pr\paren{\card{A \cap B} \geq 2c \cdot \log{N}} \leq \exp\paren{-2c \cdot \log{N}}
		\end{align*}
	and then take a union bound over all $N^{2c}$ possible choices for $A,B$ and hence argue that there should exists at least
	one collection $\SS$ such that $\card{A \cap B} < 2c \cdot \log{N}$ for any two $A,B \in \SS$. Now notice that when $r=1$, as size of each set $\SS$ is exactly $N$, 
	the collection $C_B \subseteq \SS$ and hence the previous condition on $\SS$ already satisfies the Property~(\ref{item:rnd-2}) in Definition~\ref{def:rnd}. 
\end{proof}

\subsection{Proof of Theorem~\ref{thm:dist-lower}}\label{sec:dist-lower}

To prove Theorem~\ref{thm:dist-lower} using our framework in Section~\ref{sec:dgl}, we 
parameterize the recursive family of distributions $\set{\dist_r}^c_{r=0}$ for the coverage problem, i.e., $\coverage{N,r}$, 
with the aforementioned sequences plus the packing and labeling functions which we define below. 

\textbox{Packing function $\packing_r$: \textnormal{Mapping instances $I^{i}_1,\ldots,I^{i}_{w_r}$ each over $n_{r-1}=N^{2r-1}$ elements and $m_{r-1}$ sets for any group $i \in [g_r]$ to a single
instance $I^i$ on $N^{2r}$ elements and $w_r \cdot m_{r-1}$ sets.}}{
	\begin{enumerate}
		\item Let $\AA = \set{A_1,\ldots,A_{w_r}}$ be an $(N,r,c)$-RND system with $w_r = N^{c}$ sets over some universe $\XX_i$ of $N^{2r}$ elements (guaranteed to exist by Lemma~\ref{lem:rnd-size} since $c < N$).
		By definition of $\AA$, for any set $A_j \in \AA$, $\card{A_j} = N^{2r-1} = n_{r-1}$. 
		
		\item Return the instance $I$ over the universe $\XX_i$ with the collection of all sets in $I^{i}_1,\ldots,I^{i}_{w_r}$ after mapping the elements in $I^{i}_j$ to $A_j$ arbitrarily. 
	\end{enumerate}
}

We now define the labeling family $\Labeling_r$ as a function of the index $\jstar \in [w_r]$ of special instances. 

\textbox{Labeling family $\Labeling_r$: \textnormal{Mapping instances $I^1,\ldots,I^{g_r}$ over $N^{2r}$ elements to a single instance $I$ on $n_{r}=N^{2r+1}$ elements and $m_r$ sets.}}{
	\begin{enumerate}
		\item Let $\jstar \in [w_r]$ be the index of the special instance in the distribution $\dist_r$. For each permutation $\pi$ of $[N^{2r+1}]$ we have a unique function $\labeling(\jstar,\pi)$ in the family. 
		
		\item For any instance $I^i$ for $i \in [g_r]$, map the elements in $\XX_i \setminus A_{\jstar}$ to $\pi(1,\ldots,N^{2r}-N^{2r-1})$ 
		and the elements in $A_{\jstar}$ to $\pi(N^{2r} + (g_r-1)\cdot N^{2r-1}) \ldots \pi(N^{2r} + g_r \cdot N^{2r-1}-1)$.  
		
		\item Return the instance $I$ over the universe $[N^{2r+1}]$ which consists of the collection of all sets in $I^{1},\ldots,I^{g_r}$ after the mapping above. 
	\end{enumerate}
}

Finally, we define the base case distribution $\dist_0$ of the recursive family $\set{\dist_r}_{r=0}^{c}$. 
By definition of our sequences, this distribution is over $p_0 = 1$ player, $n_0 = N$ elements, and $m_0 = 1$ set.
\textbox{Distribution $\dist_0$: \textnormal{The base case of the recursive family of distributions $\set{\dist_r}_{r=0}^{c}.$}}{

\begin{enumerate}
	\item W.p. $1/2$, the player has a single set of size $N$ covering the universe (the \Yes case). 
	\item W.p. $1/2$, the player has a single set $\set{\emptyset}$, i.e., a set that covers no elements (the \No case). 
\end{enumerate}
}

To invoke Theorem~\ref{thm:dgl}, we prove that this family is a $\gamma$-hard recursive family for the parameter $\gamma = o(r/N)$. 
The sequences clearly satisfy the required monotonicity properties. It is also straightforward to verify that $\packing_r$ and
functions $\labeling \in \Labeling_r$ are \emph{locally computable} (Definition~\ref{def:locally-computable}): both functions 
are specifying a mapping of elements to the new instance and hence each player can compute its final input by simply mapping the original input sets according to $\packing_r$ and $\labeling$ to the
new universe. In other words, the local mapping of each player $q \in P_i$ only specifies which element in the instance $I$ corresponds to which element in $I^i_j(q)$ for $j \in [w_r]$. 
It thus remains to prove the \emph{preserving} and \emph{obliviousness} property of the packing and labeling functions. 

We start by showing that the labeling family $\Labeling_r$ is oblivious. The first property of Definition~\ref{def:oblivious} is immediate to see as $\Labeling_r$ is only a function of $\jstar$ and $\sigma_r$. 
For the second property, consider any group $P_i$ and instance $I^i$; the labeling function never maps two elements belonging to a \emph{single} instance $I^i$ to the same element
in the final instance (there are however overlaps between the elements across different groups). Moreover, picking a uniformly at random labeling
function $\labeling$ from $\Labeling_r$ (as is done is $\dist_r$) results in mapping the elements in $I^i$ according to a \emph{random permutation}; as such, the set of elements 
in instance $I^i$ is mapped to a uniformly at random chosen subset of the elements in $I$, \emph{independent} of the choice of $\jstar$. As the local mapping $\labeling_q$ of each player
$q \in P_i$ is only a function of the set of elements to which elements in $I^i$ are mapped to, $\labeling_q$ is also independent of $\jstar$, proving that $\Labeling_r$ is indeed oblivious. 

The rest of this section is devoted to the proof of the preserving property of the packing and labeling functions defined for maximum coverage. 
We first make some observations about the instances created in $\dist_r$. Recall that the special instances in the distribution are $I^{1}_{\jstar},\ldots,I^{g_r}_{\jstar}$. 
After applying the packing function, each instance $I^{i}_{\jstar}$ is supported on the set of elements $A_{\jstar}$. After additionally applying the labeling function, $A_{\jstar}$ is mapped
to a unique set of elements in $I$ (according to the underlying permutation $\pi$ in $\labeling$); as a result, 

\begin{observation}\label{obs:special-unique}
	The elements in the special instances $I^{1}_{\jstar},\ldots,I^{g_r}_{\jstar}$ are mapped to \emph{disjoint} set of elements in the final instance. 
\end{observation}

The input to each player $q \in P_i$ in an instance of $\dist_{r}$ is created by mapping the sets in instances $I^i_1,\ldots,I^i_{w_r}$ (which are all sampled from distributions $\dist_{r-1}$ or $\distp_{r-1}$)
to the final instance $I$. As the packing and labeling functions, by construction, never map two elements belonging to the same instance $I^i_j$ to the same element in the final instance, the size of each 
set in the input to player $q$ is equal across any two distributions $\dist_r$ and $\dist_{r'}$ for $r \neq r'$, and thus is $N$ by definition of $\dist_0$ (we ignore empty sets in $\dist_0$ as one can consider them
as not giving any set to the player instead; these sets are only added to simplify that math). Moreover, as argued earlier, 
the elements are being mapped to the final instance according to a random permutation and hence, 
\begin{observation}\label{obs:random-set} 
	For any group $P_i$, any player $q \in P_i$, the distribution of any \emph{single} input set to player $q$ in the final instance $I \sim \dist_r$ is \emph{uniform} over
	all $N$-subsets of the universe. This also holds for an instance $I \sim \distp_r$ as marginal distribution of a player input is identical. 	
\end{observation}

We now prove the preserving property in the following two lemmas. 

\begin{lemma}\label{lem:dist-lower-alpha1}
	For any instance $I \sim \dist_r$; if $\Istar_r$ is a \Yes instance, then $I$ is also a \Yes instance. 
\end{lemma}
\begin{proof}
	
	Recall that the distribution of the special instance $\Istar_r$ is $\dist_{r-1}$. Since $\Istar_r$ is a \Yes instance, all $I^i_{\jstar}$ for $i \in [g_r]$ are also \Yes instances. 
	By definition of $\coverage{N}{r-1}$ and choice of $k_{r-1}$, this means that $\opt(I^i_{\jstar}) \geq k_{r-1} \cdot N$. 
	Moreover, by Observation~\ref{obs:special-unique}, all copies of the 
	special instance $\Istar_r$, i.e., $I^{1}_{\jstar},\ldots,I^{g_r}_{\jstar}$ are supported on disjoint set of elements in $I$. As $k_r = k_{r-1} \cdot g_r$, 
	we can pick the optimal solution from each $I^{i}_{\jstar}$ for $i \in [g_r]$ and cover at least $k_r \cdot N$ elements. By definition of 
	$\coverage{N}{r}$, this implies that $I$ is also a \Yes instance. 
\end{proof}

We now analyze the case when $\Istar_r$ is a \No instance which requires a more involved analysis.

\begin{lemma}\label{lem:dist-lower-alpha2}
For any instance $I \sim \dist_r$; if $\Istar_r$ is a \No instance, then w.p. at least $1-1/N$, $I$ is also a \No instance. 
\end{lemma}
\begin{proof}

Let $U$ be the universe of elements in $I$ and $\Ustar \subseteq U$ be the set of elements to which the elements in special instances $I^{1}_{\jstar},\ldots,I^{g_r}_{\jstar}$ are 
mapped to (these are all elements in $U$ except for the first $N^{2r}$ elements according to the permutation $\pi$ in the labeling function $\phi$).  
In the following, we bound the contribution of each set in players inputs in covering $\Ustar$ and then use the fact that $\card{U \setminus \Ustar}$ is rather small to finalize the proof.

	For any group $P_i$ for $i \in [g_r]$, let $U_i$ be the set of all elements across instances in which the players in $P_i$ are participating in. 
	Moreover, define $\Ustar_i := \Ustar \cap U_i$; notice that $\Ustar_i$ is precisely the set of elements in the special instance $I^{i}_{\jstar}$.
	We first bound the contribution of special instances.  
	
	\begin{claim}\label{clm:dist-lower-alpha2-special}
		If $\Istar_r$ is a \No instance, then for any integer $\ell \geq 0$, any collection of $\ell$ sets from the \emph{special instances} $I^{1}_{\jstar},\ldots,I^{g_r}_{\jstar}$ 
		can cover at most $k_r + \ell \cdot (2c \cdot \log{N^{2r-2}})$ elements in $\Ustar$. 
	\end{claim}
	\begin{proof}
		By definition of $\coverage{N}{r-1}$, since $\Istar_r$ is a \No instance, we have 
		$\opt(\Istar_r) \leq k_{r-1} \cdot 2c \cdot \log{(N^{2r-2})}$. 
		This implies that any collection of $\ell \geq k_{r-1}$ sets from $\Istar_r$ can only cover only $\ell \cdot 2c \cdot \log{(N^{2r-2})}$ elements; otherwise, by picking the best $k_{r-1}$ sets among this collection, we can cover more that $\opt(\Istar_r)$, 
		a contradiction. Now notice that since $\Istar_r$ is a \No instance, we know that all instances $I^{1}_{\jstar},\ldots,I^{g_r}_{\jstar}$ are also \No instances. As such, any collection of $\ell \geq k_{r-1}$ sets from each $I^{i}_{\jstar}$ can also
		cover at most $\ell \cdot 2c \cdot \log{(N^{2r-2})}$ elements from $\Ustar$. 
		
		Let $\CC$ be any collection of $\ell$ sets from special instances and $\CC_i$ be the sets in $\CC$ that are chosen from the instance $I^i_{\jstar}$. Finally, let $\ell_i = \card{\CC_i}$. 
		We have (recall that $c(\CC)$ denotes the set of covered elements by $\CC$), 
		\begin{align*}
			\card{c(\CC) \cap \Ustar} &= \sum_{i \in [g_r]} \card{c(\CC_i) \cap \Ustar_i} \leq \sum_{i \in [g_r]} (k_{r-1} + \ell_i) \cdot 2c \cdot \log{(N^{2r-2})} \\
			&= g_r \cdot k_{r-1} +  \ell \cdot 2c \cdot \log{(N^{2r-2})} \leq k_r +   \ell \cdot 2c \cdot \log{(N^{2r-2})}, 
		\end{align*}
		where the last inequality holds because $g_r \cdot k_{r-1} = k_r$. 
	\end{proof}

	We now bound the contribution of fooling instances using the RND set-systems properties. 
	
	\begin{claim}\label{clm:dist-lower-alpha2-fooling}
		With probability $1-o(1/N)$ in the instance $I$, simultaneously for all integers $\ell \geq 0$, any collection of $\ell$ sets from the \emph{fooling instances} $\set{I^i_j~\mid~i \in [g_r],~j \in [w_r] \setminus \set{\jstar}}$
		can cover at most $\ell \cdot r \cdot (2c \cdot \log{N})$ elements in $\Ustar$. 
	\end{claim}
	\begin{proof}
	Recall that for any group $i \in [g_r]$, any instance $I^i_j$ is supported on the set of elements $A_j$ in $\AA$ (before applying the
	 labeling function $\labeling$). Similarly, $\Ustar_i$ is the set $A_{\jstar}$ (again before applying $\labeling$). 
	Define $\CC_i$ as the collection of all input sets from all players in $P_i$ except the sets
	coming from the special instance. By construction, $\card{\CC_i} \leq m_{r-1} \cdot w_r \leq N^{c \cdot r}$ (as $c \geq 4r$). Moreover, for any $j \in [w_r] \setminus \set{\jstar}$, since $I^i_j \sim \distp_{r-1}$, 	
	by Observation~\ref{obs:random-set}, any member of $\CC_i$ is a set of size $N$ chosen uniformly at random from some 
	$A_{j} \neq A_{\jstar}$. This implies that $\CC_i$ satisfies the Property~(\ref{item:rnd-2}) in Definition~\ref{def:rnd} (as $\AA$ is an $(N,r,c)$-RND set-system and local mappings of elements are one to one 
	when restricted to the mapping of $\XX_i$ to $U_{i}$). As such, by definition of an RND set-system, w.p. $1-o(1/N^3)$, any set $S \in \CC$ can cover at most 
	$2c \cdot r \cdot \log{N}$ elements from $\Ustar_{i}$ and consequently $\Ustar$ as $S \cap (\Ustar \setminus \Ustar_i) = \emptyset$. 
	
	We can take a union bound over the $g_r \leq N^2$ different RND set-systems (one belonging to each group) and the above bound holds w.p. $1-o(1/N)$ for all groups simultaneously. 
	This means that any collection of $\ell$ sets across any instance $I^i_j$ for $i \in [g_r]$ and $j \neq \jstar$, can cover at most $\ell \cdot 2c \cdot r \cdot \log{N}$ elements in $\Ustar$. 
	\end{proof}
	
	In the following, we condition on the event in Claim~\ref{clm:dist-lower-alpha2-fooling}, which happens w.p. at least $1-1/N$. 
	Let $\CC = \CC_s \cup \CC_f$ be any collection of $k_r$ sets (i.e., a potential $k_r$-cover) in the input instance $I$ such that $\CC_s$ are $\CC_f$ are chosen from the special instances and fooling instances, respectively. 
	Let $\ell_s = \card{\CC_s}$ and $\ell_f = \card{\CC_f}$; we have, 
	\begin{align*}
		\card{c(\CC)} &= \card{c(\CC) \cap \Ustar} + \card{c(\CC) \cap \paren{U \setminus \Ustar}} \\
		&\leq \card{c(\CC_s) \cap \Ustar} + \card{c(\CC_f) \cap \Ustar} + \card{U \setminus \Ustar} \\
		&\leq k_r + \ell_s \cdot (2r-2) \cdot 2c \cdot \log{N} + \ell_f \cdot r \cdot 2c \cdot \log{N} + N^{2r} 
		\tag{by Claim~\ref{clm:dist-lower-alpha2-special} for the first term and Claim~\ref{clm:dist-lower-alpha2-fooling} for the second term} \\
		&\leq 4k_r + k_r \cdot (2r-2) \cdot 2c \cdot \log{N} \tag{$2k_r \geq N^{2r}$} \\
		&\leq k_r \cdot 2r \cdot 2c \cdot \log{N} \leq k_r \cdot 2c \cdot \log{N^{2r}}.
	\end{align*}
	
	This means that w.p. at least $1-1/N$, $I$ is also a \No instance. 
\end{proof}

The following claim now follows immediately from Lemmas~\ref{lem:dist-lower-alpha1} and~\ref{lem:dist-lower-alpha2}. 

\begin{claim}\label{clm:preserving-true}
	The packing function $\packing_r$ and labeling family $\Labeling_r$ defined above are $\gamma$-preserving for the parameter $\gamma = 1/N$. 
\end{claim}

We are now ready to prove Theorem~\ref{thm:dist-lower}. 
 
\begin{proof}[Proof of Theorem~\ref{thm:dist-lower}]
	The results in this section and Claim~\ref{clm:preserving-true} imply that the family of distributions $\set{\dist_r}^{c}_{r=0}$ for the $\coverage{N,r}$ are $\gamma$-hard for the
	parameter $\gamma = 1/N$, as long as $r \leq 4c \leq 4\sqrt{N/12}$. Consequently, by Theorem~\ref{thm:dgl}, any $r$-round protocol that can compute the value of $\coverage{N,r}$ on $\dist_r$ w.p. at
	least $2/3 + r \cdot \gamma = 2/3 + r/N < 3/4$ requires $\Omega(w_r/r^4) = \Omega(N^c/r^4)$ total communication. 
	Recall that the gap between the value of optimal solution between \Yes and \No instances of $\coverage{N}{r}$ is at least $N/ \paren{2c \cdot \log{(N^{2r})}} \geq {\paren{\frac{k_r^{1/2r}}{2c\cdot\log{k_r}}}}$. 
	As such, any $r$-round distributed algorithm that can approximate the value of optimal solution to within a factor better than this w.p. at least $3/4$
	can distinguish between \Yes and \No cases of this distribution, and hence requires $ \Omega(N^{c-2r}/r^4) = \Omega\paren{\frac{k_r}{r^4} \cdot m^{\frac{c}{(c+2)\cdot 4r}}}$ per player communication. 
	Finally, since $N \leq 2k_r^{1/2r}$, the condition $c \leq \sqrt{N/12}$ holds as long as $c =o\paren{\frac{\log{k_r}}{\log\log{k_r}}}$, finalizing the proof. 
\end{proof}

\renewcommand{\CC}{\ensuremath{\mathcal{C}}}

\newcommand{\Sstar}{\ensuremath{S^{\star}}}

\newcommand{\SPGreedy}{\ensuremath{\textnormal{\textsf{SPGreedy}}}\xspace}

\newcommand{\ISGreedy}{\ensuremath{\textnormal{\textsf{ISGreedy}}}\xspace}

\newcommand{\GreedySketch}{\ensuremath{\textnormal{\textsf{GreedySketch}}}\xspace}

\newcommand{\optp}{\ensuremath{\widetilde{\opt}}}

\section{Distributed Algorithms for Maximum Coverage} \label{sec:dist-algs}

In this section, we show that both the \emph{round-approximation} tradeoff and the \emph{round-communication} tradeoff achieved by our lower bound in Theorem~\ref{thm:dist-lower} are essentially
tight, formalizing Result~\ref{res:dist-upper}. 

\subsection{An $O(r \cdot k^{1/r})$-Approximation Algorithm} \label{sec:dist-k}

Recall that Theorem~\ref{thm:dist-lower} shows that getting better than $k^{\Omega(1/r)}$ approximation in $r$ rounds requires a relatively large
communication of $m^{\Omega(1/r)}$, (potentially) larger than any $\poly(n)$. In this section, we prove that this round-approximation tradeoff is essentially tight by showing that one can always
obtain a $k^{O(1/r)}$ approximation (with a slightly larger constant in the exponent) in $r$ rounds using a limited communication of nearly linear in $n$. 

\begin{theorem}\label{thm:dist-upper-k}
	There exists a deterministic distributed algorithm for the maximum coverage problem that for any integer $r \geq 1$ computes 
	an $O(r \cdot k^{1/r+1})$ approximation in $r$ rounds and $\Ot(n)$ communication per each machine. 
\end{theorem}

On a high level, our algorithm follows an iterative sketching method: in each round, each machine computes a small collection $\CC_i$ of its input sets $\SS_i$ as a sketch and sends
it to the coordinator. The coordinator is maintaining a collection of sets $\XX$ and updates it by iterating over the received sketches and picking any set that still has a relatively large contribution to this partial solution. 
The coordinator then communicates the set of elements covered by $\XX$ to the machines and the machines update their inputs accordingly and repeat this process. At the end,
the coordinator returns (a constant approximation to) the optimal $k$-cover over the collection of \emph{all} received sets across different rounds. 

In the following, we assume that our algorithm is given a value $\optp$ such that $\opt \leq \optp \leq 2 \cdot \opt$. 
We can remove this assumption by guessing the value of $\optp$ in powers of two (up to $n$) and solve the problem simultaneously for all of them and return the best solution, which 
increases the communication cost by only an $O(\log{n})$ factor. 

We first introduce the algorithm for computing the sketch on each machine; the algorithm is a simple thresholding version of the greedy algorithm for maximum coverage.

\textbox{\textnormal{$\GreedySketch(U,\SS,\tau)$. An algorithm for computing the sketch of each machine's input.}}{
	
	\smallskip
	
	\textbf{Input:} A collection $\SS$ of sets from $[n]$, a target universe $U \subseteq [n]$, and a threshold $\tau$. \\
	\textbf{Output:} A collection $\CC$ of subsets of $U$. 

	\begin{enumerate}
		\item Let $\CC = \emptyset$ initially. 
		\item\label{line:greedysketch-add} Iterate over the sets in $\SS$ in an arbitrary order and for each set $S \in \SS$, if $\card{(S \cap U) \setminus c(\CC)} \geq \tau$, then add $(S \cap U) \setminus c(\CC)$ to $\CC$. 
		\item Return $\CC$ as the answer. 
	\end{enumerate}
}

Notice that in the Line~(\ref{line:greedysketch-add}) 
of \GreedySketch, we are adding the \emph{new contribution} of the set $S$ and not the complete set itself. This way, we can bound the total representation size of the output collection $\CC$ 
by $\Ot(n)$ (as each element in $U$ appears in at most one set). We now present our algorithm in Theorem~\ref{thm:dist-upper-k}. 

\textbox{Algorithm 2: \textnormal{Iterative Sketching Greedy (\ISGreedy).}}{

\medskip

\textbf{Input:} A collection $\SS_i$ of subsets of $[n]$ for each machine $i \in [p]$ and a value $\optp \in [\opt,2\cdot\opt]$.   

\smallskip

\textbf{Output:} A $k$-cover from the sets in $\SS := \bigcup_{i \in [p]} \SS_i$. 

\begin{enumerate}
	\item Let $\XX^0 = \emptyset$ and $U^0_i = [n]$, for each $i \in [p]$ initially. Define $\tau := \optp / 4r \cdot k$.   
	\item For $j = 1$ to $r$ rounds: 
	\begin{enumerate}
		\item\label{line:isgreedy-sketch} Each machine $i$ computes $\CC^j_i = \GreedySketch(U^{j-1}_i,\SS_i,\tau)$ and sends it to coordinator. 
		\item\label{line:isgreedy-add} The coordinator sets $\XX^j = \XX^{j-1}$ initially and iterates over the sets in $\bigcup_{i \in [p]}\CC^j_i$, in decreasing order of $\card{c(\CC^j_i)}$ over $i$ (and consistent
		 with the order in \GreedySketch for each particular $i$), and adds each set $S$ to $\XX^j$ if $\card{S \setminus c(\XX^j)} \geq \frac{1}{k^{1/r+1}} \cdot \card{S}$. 
		\item\label{line:isgreedy-update} The coordinator communicates $c(\XX^j)$ to each machine $i$ and the machine updates its input by setting $U^j_i = c(\CC^j_i) \setminus c(\XX^j)$. 
	\end{enumerate}
	\item At the end, the coordinator returns the best $k$-cover among all sets in $\CC := \bigcup_{i \in [p], j \in [r]} \CC^j_i$ sent by the machines over all rounds. 
\end{enumerate}
}

The round complexity of \ISGreedy is trivially $r$. For its communication cost, notice that at each round, each machine is communicating at most $\Ot(n)$ bits and the coordinator communicates
$\Ot(n)$ bits back to each machine. As the number of rounds never needs to be more than $O(\log{k})$, we obtain that \ISGreedy requires $\Ot(n)$ communication per each machine. 
Therefore, it only remains to analyze the approximation guarantee of this algorithm. To do so, it suffices to show that, 

\begin{lemma}\label{lem:ub-k-apx}
	Define $\CC := \bigcup_{i \in [p], j \in [r]} \CC^j_i$. The optimal $k$-cover of $\CC$ covers $\paren{\frac{\opt}{4r \cdot k^{1/r+1}}}$ elements. 
\end{lemma}
\begin{proof}

We prove Lemma~\ref{lem:ub-k-apx} by analyzing multiple cases. We start with an easy case when $\card{\XX^r} \geq k$. 

\begin{claim}\label{clm:ub-xx-k}
	If $\card{\XX^r} \geq k$, then the optimal $k$-cover of $\XX^r \subseteq \CC$ covers $\paren{\frac{\opt}{4r \cdot k^{1/r+1}}}$ elements. 
\end{claim}
\begin{proof}
	Consider the first $k$ sets added to the collection $\XX^r$. Any set $S$ that is added to $\XX^r$ in (Line~(\ref{line:isgreedy-add}) of \ISGreedy) covers $\frac{1}{k^{1/r+1}} \cdot \card{S}$ new elements. 
	Moreover, $\card{S} \geq \tau = \optp/4rk$ (by Line~(\ref{line:greedysketch-add}) of the \GreedySketch algorithm). Hence, the first $k$ sets added to $\XX^r$ already cover at least, 
	 \[
	 	k \cdot \frac{1}{k^{1/r+1}} \cdot \frac{\optp}{4rk} \geq \frac{\opt}{4r \cdot k^{1/r+1}} 
	 \] 
	 elements, proving the claim. 
\end{proof}

The more involved case is when $\card{\XX^r} < k$, which we analyze below. Recall that $\CC^j_i$ is the collection computed by $\GreedySketch(U^{j-1}_i,\SS_i,\tau)$ on the
machine $i \in [p]$ in round $j$. We can assume that each $\card{\CC^j_i} < k$; otherwise consider the
smallest value of $j$ for which the for the first time there exists an $i \in [p]$ with $\card{\CC^j_i} \geq k$ (if for this value of $j$, there are more than one choice for $i$ choose the one with the largest size of $c(\CC^j_i)$): in 
Line~(\ref{line:isgreedy-add}), the coordinator would add all the sets in ${\CC^j_i}$ to $\XX^j$ making $\card{\XX^j} \geq k$, a contradiction with the assumption that $\card{\XX^r} < k$.

By the argument above, if there exists a machine $i \in [p]$, with $\card{c(\CC^1_i)} > \opt/4k^{1/r+1}$, we are already done. 
This is because the collection $\CC^1_i$ contains at most $k$ sets and hence $\CC^1_i$ is a valid $k$-cover in $\CC$ that covers $(\opt/4k^{1/r+1})$ elements, proving the lemma 
in this case. It remains to analyze the more involved case when none of the above happens.  

\begin{lemma}\label{lem:ub-k-hard}
	Suppose $\card{\XX^r} < k$ and $\card{c(\CC^1_i)} \leq \opt/4k^{1/r+1}$ for all $i \in [p]$; then, the optimal $k$-cover of $\CC$ covers $\paren{\frac{\opt}{4r \cdot k^{1/r+1}}}$ elements. 
\end{lemma}
\begin{proof}
	Recall that in each round $j \in [r]$, each machine $i \in [p]$ first computes a collection $\CC^{j}_i$ from the universe $U^{j-1}_i$ as its sketch (using \GreedySketch) and sends
	it to the coordinator; at the end of the round also this machine $i$ updates its target universe for the next round to $U^{j}_i \subseteq \CC^{j}_i$. 
	We first show that this target universe $U^{j}_i$ shrinks in each round by a large factor compared to $\CC^{j}_i$. 
	 \begin{claim}\label{clm:ub-k-apx-shrink}
	 	For any round $j \in [r]$ and any machine $i \in [p]$, $\card{U^j_i} \leq \paren{1/k^{1/r+1}} \cdot \card{c(\CC^{j}_i)}$. 
	 \end{claim}
	 \begin{proof}
	 	Consider any $i \in [p]$ and round $j \in [r]$;  by Line~(\ref{line:isgreedy-update}) of \ISGreedy, we know $U^j_i = c(\CC^j_i) \setminus c(\XX^j)$. Hence,
		it suffices to show that $\XX^j$ covers $(1-1/k^{1/r+1})$ fraction of $c(\CC^j_i)$. This is true because for any set $S \in \CC^j_i$ that is \emph{not} added to $\XX^j$, we have, 
		$\card{S \setminus c(\XX^j)} < \frac{1}{k^{1/r+1}} \cdot \card{S}$, meaning that at most $1/k^{1/r+1}$ fraction of any set $S \in \CC^j_i$ can remain uncovered by $\XX^j$ at the end of the round $j$. 
	 \end{proof}

	  By Claim~\ref{clm:ub-k-apx-shrink}, and the assumption on size of $\card{c(\CC^1_i)}$ in the lemma statement, we have, 
	  \begin{align}
		  \card{c(\CC^{r}_i)} &\leq \card{U^{r-1}_i} \leq \paren{\frac{1}{k^{1/r+1}}} \cdot \card{c(\CC^{r-1}_i)} \leq \paren{\frac{1}{k^{1/r+1}}} \cdot \card{U^{r-2}_i}  
		  \tag{since $U^{j}_i \subseteq c(\CC^j_i) \subseteq U^{j-1}_i$ by construction of \ISGreedy and  \GreedySketch} \\
		  &\leq \paren{\frac{1}{k^{1/r+1}}}^{r-1} \cdot \card{c(\CC^{1}_i)} \leq \paren{\frac{1}{k^{1/r+1}}}^{r-1} \cdot \frac{\opt}{4k^{1/r+1}} 
		  \tag{by expanding the bound on each $\card{U^{j}_i}$ recursively and using the bound on $\card{c(C^1_i)}$} \\  
		  &\leq \frac{\opt}{4k^{r/r+1}} \label{eq:ub-last-round}. 
	  \end{align}
	  
	  Fix any optimal solution $\OPT$. We make the sets in $\OPT$ \emph{disjoint} by arbitrarily assigning each element in $c(\OPT)$ to exactly one of the sets
	  that contains it. Hence, a set $O \in \OPT$ is a subset of one of the original sets in $\SS$; we slightly abuse the notation and say $O$ belongs to $\SS$ (or input of some machine)
	  to mean that the corresponding super set belongs to $\SS$.  In the following, we use Eq~(\ref{eq:ub-last-round}) to argue that any set $O \in \OPT$ has 
	  a ``good representative'' in the collection $\CC$. This is the key part of the proof of Lemma~\ref{lem:ub-k-hard} and the next two claims 
	  are dedicated to its proof. 
	  
	  We first show that for any set $O$ in the optimal solution that belonged to machine $i \in [p]$,  if $O$ was never picked in any $\XX^j$ during
	  the algorithm, then the universe $U^{j}_i$ at any step covers a large portion of $O$. For any $j \in [r]$ and $i \in [p]$, define $X^j := c(\XX^j)$ and $C^{j}_i = c(\CC^{j}_i)$. We have, 
	  
	  \begin{claim}\label{clm:ub-k-apx-intersection-1}
	  		For any set $O \in \OPT \setminus \CC$ and the parameter $\tau$ defined in \ISGreedy, if $O$ appears in the input of machine $i \in [p]$, then, for any $j \in [r]$, 
			\[
				\card{O \cap U^{j}_i} \geq {\card{O \setminus X^j} - j \cdot \tau}. 
			\]
	  \end{claim}
	  \begin{proof}
	  	The idea behind the proof is as follows. In each round $j$, among the elements already in $U^{j-1}_i$, at most $\tau$ elements of $O$ can be left uncovered by the 
		set $C^{j}_i$ as otherwise the \GreedySketch algorithm should have picked $O$ (a contradiction with $O \notin \CC$). Moreover, any element in $C^{j}_i$ but not $U^{j}_i$ 
		is covered by $c(\XX^{j})$ i.e., $X^{j}$ and hence can be accounted for in the term $\card{O \setminus X^j}$. 
		
		We now formalize the proof. The proof is by induction. The base case for $j = 0$ is trivially true as $U^0_i = [n]$ and $X^0 = \emptyset$ (as $\XX^0 = \emptyset$). 
		Now assume inductively that this is the case for integers smaller than $j$ and we prove it 
		for $j$. By Line~(\ref{line:greedysketch-add}) of \GreedySketch, we know $\card{O \cap U^{j-1}_i \setminus C^j_i} < \tau$ 
		as otherwise the set $O$ would have been picked by $\GreedySketch(U_i^{j-1},\SS_i,\tau)$ in \ISGreedy, a contradiction with
		the fact that $O \notin \CC$. Using this plus the fact that $C^j_i = c(\CC^{j}_i) \subseteq U^{j-1}_i$, we have, 
		\begin{align}
			\card{O \cap C^j_i} = \card{O \cap C^j_i \cap U^{j-1}_i} \geq \card{O \cap U^{j-1}_i} - \card{O \cap U^{j-1}_i \setminus C^j_i} \geq \card{O \setminus X^{j-1}} - j \cdot \tau, \label{eq:ub-k-induction}
		\end{align}
		where the last inequality is by induction hypothesis on the first term and the bound of $\tau$ on the second term. 
		
		To continue, define $Y^j = X^{j} \setminus X^{j-1}$, i.e., the set of new elements covered by $\XX^j$ compared to $\XX^{j-1}$. By construction of the algorithm \ISGreedy,
		$U^{j}_i = C^{j}_i \setminus X^{j} = C^{j}_i \setminus Y^j$ as $U^{j-1}_i$ and consequently $C^{j}_i$ do not have any intersection with $X^{j-1}$. We now have, 	
		\begin{align*}
			\card{O \cap U^j_i} &= \card{O \cap \paren{C^{j}_i \setminus Y^j}} \geq \card{O \cap C^{j}_i} - \card{O \cap Y^j} \\
			&\Geq{Eq~(\ref{eq:ub-k-induction})} \card{O \setminus X^{j-1}} - j \cdot \tau - \card{O \cap Y^j} \\
			&= \card{O \setminus \paren{X^j \setminus Y^j}} - j\cdot \tau - \card{O \cap Y^j} \tag{by definition of $Y^j = X^j \setminus X^{j-1}$} \\
			&= \card{O \setminus X^j} - j \cdot \tau \tag{since $Y_j \subseteq X_j$},
		\end{align*}
		which proves the induction step. 
	  \end{proof}
	  
	  We next  argue that since any set $O \in \OPT$ that is located on machine $i$ is ``well represented'' in $U^{r}_i$ by Claim~\ref{clm:ub-k-apx-intersection-1} (if not already picked
	  in $\XX^r$), and since by Eq~(\ref{eq:ub-last-round}), 
	  size of $\CC^{r}_i$ and consequently the number of sets sent by machine $i$ in $\CC^{r}_i$ is small, there should exists a set in $\CC^{r}_i$ that also represents $O$ rather closely. Formally, 
	  
	  \begin{claim}\label{clm:ub-k-apx-intersection-2}
	  		For any set $O \in \OPT$, there exists a set $S_O \in \CC$ such that for the parameter $\tau$ defined in \ISGreedy,
			\[
				\card{O \cap S_O} \geq \frac{\card{O \setminus X^r} - r \cdot \tau}{r \cdot k^{1/r+1}}. 
			\]
	  \end{claim}
	  \begin{proof}
	  	Fix a set $O \in \OPT$ and assume it appears in the input of machine $i \in [p]$. The claim is trivially true if $O \in \CC$ (as we can take $S_O = O$). Hence, assume $O \in \OPT \setminus \CC$. 
		By Claim~\ref{clm:ub-k-apx-intersection-1} and the fact that $U^{r}_i \subseteq C^{r}_i$, at the end of the last round $r$, we have, 
		\begin{align*}
			\card{O \cap C^{r}_i} \geq \card{O \cap U^{r}_i} \Geq{Claim~\ref{clm:ub-k-apx-intersection-1}} {\card{O \setminus X^{r}} - r \cdot \tau}.
		\end{align*}
		Moreover, by Eq~(\ref{eq:ub-last-round}), $\card{C^r_i} \leq \opt/4k^{r/r+1}$. Since any set added to $\CC^{r}_i$ increases $C^{r}_i = c(\CC^{r}_i)$ by at least $\tau = \opt/4kr$ elements (by construction
		of \GreedySketch), we know that, 
		\begin{align*}
			\card{\CC^{r}_i} \leq \frac{\card{C^{r}_i}}{\opt/4kr} \Leq{Eq~(\ref{eq:ub-last-round})} r \cdot k^{1/r+1}.
		\end{align*}
		It is easy to see that there exists a set $S_O \in \CC^{r}_i$ that covers at least $1/\card{\CC^{r}_i}$ fraction of $O \cap C^{r}_i$; combining this with the equations above, 
		we obtain that, 
		\[
				\card{O \cap S_O} \geq \frac{\card{O \setminus X^r} - r \cdot \tau}{r \cdot k^{1/r+1}}. \qed
		\]
		
	  \end{proof}
	  
	  We are now ready to finalize the proof of Lemma~\ref{lem:ub-k-hard}. Define $\CC_O := \set{S_O \in \CC \mid O \in \OPT}$ for the sets $S_O$ defined in Claim~\ref{clm:ub-k-apx-intersection-2}. 
	  Clearly, $\CC_O \subseteq \CC$ and $\card{\CC_O} \leq k$.  Additionally, recall that $\card{\XX} < k$ by the assumption in the lemma statement. Consequently, both $\CC_O$ and 
	  $\XX$ are $k$-covers in $\CC$. In the following, we show that the best of these two collections covers $(\opt/4r\cdot k^{1/r+1})$ elements. 
	  \begin{align*}
	  	\card{c(\CC_O)} + \card{c(\XX^r)} &= \card{\bigcup_{O \in \OPT}{S_O}}+ \card{X^r} \geq \card{\bigcup_{O \in \OPT}\paren{O \cap S_O}}+ \card{X^r} \\
		&= \sum_{O \in \OPT} \card{O \cap S_O} + \card{X^r} 
		\tag{as by the discussion before Claim~\ref{clm:ub-k-apx-intersection-1} we assume the sets in $\OPT$ are disjoint} \\
		&\Geq{Claim~\ref{clm:ub-k-apx-intersection-2}} \sum_{O \in \OPT} \paren{\frac{\card{O \setminus X^r} - r \cdot \tau}{r \cdot k^{1/r+1}}} + \card{X^r} \\
		&= \frac{\card{\bigcup_{O \in \OPT} O \setminus X^r} - k \cdot r \cdot \tau}{r \cdot k^{1/r+1}} + \card{X^r} 
		\tag{again by the assumption on the disjointness of the sets in $\OPT$ and the fact that $\card{\OPT} = k$} \\
		&\geq \frac{\card{c(\OPT)} -\card{X^r} - \optp/4}{r \cdot k^{1/r+1}}  + \card{X^r} 
		\tag{as $\tau = \optp/4kr$} \\ 
		&\geq \frac{\card{c(\OPT)} - \opt/2}{r \cdot k^{1/r+1}} \geq {\frac{\opt}{2r \cdot k^{1/r+1}}}.
		\tag{as $\card{c(\OPT)} = \opt$ and $\optp \leq 2 \cdot \opt$}
	  \end{align*}
	  As a result, at least one of $\CC_O$ or $\XX^r$ is a $k$-cover that covers $(\opt/4r\cdot k^{1/r+1})$ elements, finalizing the proof. 
\end{proof}
	  Lemma~\ref{lem:ub-k-apx} now follows immediately from Claim~\ref{clm:ub-xx-k} and Lemma~\ref{lem:ub-k-apx}. 
\end{proof}

Theorem~\ref{thm:dist-upper-k} follows from Lemma~\ref{lem:ub-k-apx} as the coordinator can simply run any constant factor approximation algorithm for maximum coverage on the collection $\CC$ and obtains
the final result.

\subsection{An $(\frac{e}{e-1})$-Approximation Algorithm} \label{sec:dist-m}

We now prove that the \emph{round-communication} tradeoff for the distributed maximum coverage problem proven in Theorem~\ref{thm:dist-lower} is essentially 
tight. Theorem~\ref{thm:dist-lower} shows that using $k \cdot m^{O(1/r)}$ communication in $r$ rounds only allows for a relatively large
approximation factor of $k^{\Omega(1/r)}$. Here, we show that we can always obtain an (almost) $\paren{\frac{e}{e-1}}$-approximation (the optimal approximation ratio with 
sublinear in $m$ communication) in $r$ rounds using $k \cdot m^{\Omega(1/r)}$ (for some larger constant in the exponent). 

As stated in the introduction, our algorithm in this part is quite general and works for maximizing any monotone submodular function subject to a cardinality constraint (see Appendix~\ref{app:submodular} for definitions). 
Hence, in the following, we present our results in this more general form. 

\begin{theorem}\label{thm:dist-upper-e}
	There exists a randomized distributed algorithm for submodular maximization subject to cardinality constraint that for any ground set $V$ of size $m$, any monotone submodular function $f: 2^V \rightarrow \IR^{+}$, and any
	integer $r \geq 1$ and parameter $\eps \in (0,1)$, with high probability computes an $\paren{\frac{e}{e-1} + \eps}$-approximation in $r$ rounds while 
	communicating $O(k \cdot m^{O(1/\eps \cdot r)})$ items from $V$.
\end{theorem}

\medskip
\noindent
\textbf{Remark:} We emphasize that the interesting case in Theorem~\ref{thm:dist-upper-e} is when $r = \Omega(1/\eps)$; otherwise, the communication cost guarantee of Theorem~\ref{thm:dist-upper-e} can be achieved by a trivial protocol that
communicates the whole input to the coordinator in just a single round. Consequently, in the following, we always assume that $r = \Omega(1/\eps)$.  

\medskip

Our algorithm follows the sample-and-prune technique of~\cite{KumarMVV13}. At each round, we sample a set of items from the machines and send them
to the coordinator. The coordinator then computes a greedy solution $X$ over the received sets and reports $X$ back to the machines. The machines then prune any item that 
cannot be added to this partial greedy solution $X$ and continue this process in the next rounds. At the end, the coordinator outputs $X$. By using a thresholding greedy algorithm
and a more careful analysis, we show that the dependence of the number of rounds on $\Omega(\log{\Delta})$ (where $\Delta$ is the ratio of maximum value of $f$ on any singleton set to its minimum
value) in~\cite{KumarMVV13} can be completely avoided, resulting in an algorithm with only constant number of rounds. 

We assume that the algorithm is given a value $\optp$ such that $\opt \leq \optp \leq 2 \cdot \opt$. In general, one can guess $\optp$ in powers of two in the range $\Delta$ to $k \cdot \Delta$
in parallel and solve the problem for all of them and return the best solution. This would increase the communication cost by only a factor of $\Theta(\log{k})$ (and one extra round of communication just to communicate $\Delta$ if it 
is unknown).  We now present our algorithm. 

\clearpage 

\textbox{Algorithm 1: \textnormal{Sample and Prune Greedy (\SPGreedy).}}{

\medskip

\textbf{Input:} A collection $V_i \subseteq V$ of items for each machine $i \in [p]$ and a value $\optp \in [\opt,2\cdot\opt]$.  

\smallskip

\textbf{Output:} A collection of $k$ items from $V$.  

\begin{enumerate} 
	\item Define the parameters $\ell:= \ceil{\lg_{(1+\eps)}{(2e)}} (= \Theta(1/\eps))$ and $s = \ceil{r/\ell}$. The algorithm consists of $\ell$ \emph{iterations} each with $s$ \emph{steps}. 
	
	\item\label{line:spgreedy-iteration} For $j = 1$ to $\ell$ iterations: 
	\begin{enumerate}
		\item Let $\tau_j = \frac{\optp}{k} \cdot \paren{\frac{1}{1+\eps}}^{j-1}$ and $X^{j,0} = X^{j-1,s}$ initially (we assume $X^{0,*} = \emptyset$). 
		\item\label{line:spgreedy-step} For $t = 1$ to $s$ steps: 
		\begin{enumerate}[(i)]
				\item\label{line:spgreedy-set-defined} Define $V^{j,t} = \set{a \in V \mid f_{X^{j,(t-1)}}(a) \geq \tau_j}$. 
				\item Each machine $i \in [p]$ samples each item in $V^{j,t} \cap V_i$ independently and with probability 
				$
					q_t := \begin{cases} 
					\frac{4k\log{m}}{m^{1-(t/s)}}~~~~&\text{if } t < s \\ 
					1 & \text{if } t=s
					\end{cases},$
				 and sends them to the coordinator.  
				\item\label{line:spgreedy-added} The coordinator iterates over each received item $a$ (in an arbitrary order) and adds $a$ to $X^{j,t}$ iff ${f_{X^{j,t}}(a)} \geq \tau_j$. 
				\item The coordinator communicates the set $X^{j,t}$ to the machines. 
		\end{enumerate} 
	\end{enumerate} 
	\item The coordinator returns $X^{\ell,s}$ in the last step (if at any earlier point of the algorithm size of some $X^{*,*}$ is already $k$, the coordinator terminates the algorithm and outputs this set as the answer).
	\end{enumerate}
}

\SPGreedy requires $\ell = \Theta(1/\eps)$ iterations each consists of $s = \ceil{r/\ell}$ steps. Moreover, each step can be implemented in one round of communication. 
As such, the round complexity of this algorithm is simply $O(r)$ (as we assumed $r = \Omega(1/\eps)$). In the following, we prove a bound on the communication cost of this algorithm and then analyze its approximation guarantee. To do so, we need the following auxiliary lemma on the size
of each set $V^{j,t}$ in the algorithm. 

\begin{lemma}\label{lem:ub-e-set-cc}
		For any $j \in [\ell]$ and any $t \in [s]$, $\card{V^{j,t}} \leq m^{1-(t-1)/s}$ w.p. at least $1-1/m^{2k}$. 
\end{lemma}
\begin{proof}
	Fix any iteration $j \in [\ell]$ and observe that $X^{j,0} \subseteq \ldots \subseteq X^{j,s}$. By submodularity of $f(\cdot)$, this implies that for any $a \in V$, 
	$f_{X^{j,0}}(a) \geq \ldots \geq f_{X^{j,s}}(a)$ and hence $V^{j,1} \supseteq V^{j,2} \supseteq \ldots \supseteq V^{j,s}$. 
	
	The bound in the lemma statement is trivially true for $t = 1$; hence, we prove it for any $t > 1$. To do so, we show that the collection $X^{j,t-1}$, computed at the end of the $(t-1)$-th step
	in iteration $j$, has the property that the corresponding collection $V^{j,t}$ (which is uniquely identified by $X^{j,t}$) has its size bounded as in the lemma statement. 
	
	Fix any set $A$ of up to $k$ items from $V$. We say that $A$ is \emph{bad} iff the set $V_A := \set{a \in V \mid f_A(a) \geq \tau_j}$ has 
	size more than $m^{1-(t-1)/s}$. For the set $X^{j,t}$ to be equal to $A$ at the end of the $(t-1)$-th step (in iteration $j$), necessarily no item from $V_A$ should be sampled
	by any of the machines in that round. As such, for any bad set $A \subseteq V$, 
	\begin{align*}
		\Pr\paren{X^{j,t} = A} \leq \paren{1-q_{t-1}}^{\card{V_A}} \leq \paren{1-\frac{4k\cdot\log{m}}{m^{1-(t-1)/s}}}^{m^{1-(t-1)/s}} \leq \exp\paren{-4k\cdot \log{m}} \leq m^{-4k}.
	\end{align*}
	Taking a union bound over $\sum_{i=1}^{k} {{m} \choose i} = O(m^{k})$ possible choices for a bad set $A$, the probability that any bad set $A$ is chosen as the set $X^{j,t}$ is 
	smaller than $1/m^{3k}$. Conditioned on this event, the set $V_{X^{j,t}}$ for the next round, i.e., the $t$-th round, has size at most $m^{1-(t-1)/s}$. Taking a union bound over all $j \in [\ell]$ and $t \in [s]$ 
	finalizes the proof. 
\end{proof}

It is now easy to bound the communication cost of this protocol. 

\begin{lemma}\label{lem:ub-e-cc}
	\SPGreedy communicates at most $O(r \cdot k \cdot m^{1/s} \cdot \log{m})$ items w.p. at least $1-1/m^{k}$.
\end{lemma}
\begin{proof}
		We condition on the event in Lemma~\ref{lem:ub-e-set-cc}. As such, for each iteration $j \in [\ell]$ and each step $t \in [s]$ in this iteration, $V^{j,t}$ is of size $m^{1-(t-1)/s}$ at most. 
		Consequently, the total number of items sampled by the machines in step $t$ is in expectation at most $m^{1-(t-1)/s} \cdot q_t = m^{1/s} \cdot 4k\log{m}$. This means that, by Chernoff
		bound, w.p. at least $1-1/m^{2k}$, at most $O(m^{1/s} \cdot k\log{m})$ items are communicated by each machine in this step. The coordinator also communicates at most $k$ items 
		to each machine in each step. The bound in the lemma statement now follows by taking a union bound over all $O(r)$ iterations and steps. 
\end{proof}

We now prove a bound on the quality of the returned solution. 

\begin{lemma}\label{lem:ub-e-apx}
	Suppose $X$ is the set returned by $\SPGreedy$; then, $f(X) \geq (1-1/e-\eps) \cdot \opt$. 
\end{lemma}
\begin{proof}
	We first argue that if the set $X$ has size $<k$ then $f(X) \geq (1-1/e) \cdot \opt$ already; note that in this case, $X = X^{\ell,s}$. 
	Let $\OPT$ be an optimal solution and consider any item $o \in \OPT$ that was never picked by the coordinator to be added to $X$; this in particular 
	means that $o$ was not added to $X^{\ell,s}$ which implies,
	 \begin{align}
	 f_X(o) = f_{X^{\ell,s}}(o) < \tau_{\ell} = \frac{\optp}{k} \cdot \paren{\frac{1}{1+\eps}}^{\ell} = \frac{\optp}{2e \cdot k} \leq \frac{\opt}{e\cdot k} \label{eq:ub-m-small-items}.
	\end{align}
	The first inequality in Eq~(\ref{eq:ub-m-small-items}) holds because in step $s$ of each iteration, \emph{every} item $a \in V$ with $f_{X^{\ell,s-1}}(a) \geq f_{X^{\ell,s}}(a)$ (by submodularity) is sent to the coordinator and hence 
	if $f_{X^{\ell,s}}(o) \geq \tau_{\ell}$ the coordinator would be able to find it and add it to $X^{\ell,s}$. 
	The next two equalities are by the choices of $\tau_{\ell}$ and $\ell$, respectively, and the last inequality is true since $\optp \leq 2\opt$. Using this bound
	and the monotone submodularity of $f(\cdot)$, we can write, 
	\begin{align*}
		f(\OPT) \Leq{Fact~\ref{fact:submodular-greedy}} f(X) + \sum_{o \in \OPT \setminus X} f_X(o) \Leq{Eq~(\ref{eq:ub-m-small-items})} f(X) + \card{\OPT}\cdot \frac{\opt}{e \cdot k} = f(X) + \opt/e
	\end{align*}
	as $\card{\OPT} = k$, which finalizes the proof in this case. 
	
	We now consider the more involved case where the coordinator picks exactly $k$ items in $X$. To continue, we need the following definitions. 
	Let $x_1,\ldots,x_k$ be the items added to $X$ by the coordinator in this particular order. For any $i \in [k]$, define $X^{<i} = x_1,\ldots,x_{i-1}$, i.e., the first $i-1$ items added to $X$ (define $X^{<1} = \emptyset$).
	We have, 
	\begin{claim}\label{clm:ub-e-apx-contribution}
		For any $i \in [k]$, 
		\begin{align*}
			f_{X^{<i}}(x_i) \geq \frac{f(\OPT) - f(X^{<i})}{(1+\eps) \cdot k}.
		\end{align*}
	\end{claim}
	\begin{proof}
		For any item $x_i$ for $i \in [k]$, by construction of \SPGreedy, if $i$ is added in iteration $j \in [\ell]$ to $X$, then, 
		\begin{align}
			f_{X^{<i}}(x_i) \geq \tau_j \label{eq:ub-m-added}.
		\end{align}
		
		Suppose first that the item $x_i$ is added to $X$ in the first iteration. By the above equation, 
		\begin{align*}
			f_{X^{<i}}(x_i) \Geq{Eq~(\ref{eq:ub-m-added})} \tau_1 = \frac{\optp}{(1+\eps) \cdot k} \geq \frac{\opt}{(1+\eps) \cdot k},
		\end{align*}
		by the bounds on $\tau_1$ and $\optp$. This proves the lemma for any item $x_i$ that is added to $X$ in the first iteration. 
		Now suppose $x_i$ is added in the iteration $j > 1$. 
		
		Consider the item $o^* \in \OPT \setminus X^{<i}$ with the maximum marginal contribution to $f_{X^{<i}}$. Recall that since $f(\cdot)$ is submodular, by Fact~\ref{fact:submodular-subadditive}, $f_{X^{<i}}(\cdot)$ is 
		subadditive. We have,
		\begin{align}
			f_{X^{<i}}(o^*) &= \max_{o \in \OPT \setminus X^{<i}} f_{X^{<i}}(o) 
			\Geq{Fact~\ref{fact:submodular-subadditive}} \frac{1}{k} \cdot \sum_{o \in \OPT \setminus X^{<i}} f_{X^{<i}}(o) 
			\Geq{Fact~\ref{fact:submodular-greedy}} \frac{1}{k} \cdot \paren{f(\OPT) - f(X^{<i})} \label{eq:ub-m-o*}
		\end{align}
		On the other hand, we also know that $o^*$ does not belong to $X^{<i}$, meaning that it was not added to $X^{<i}$ at least by end of iteration $j-1$ (since $x_i$ is added to $X^{<i}$ in iteration $j$). 
		Hence, again by construction of \SPGreedy, similar to the case in Eq~(\ref{eq:ub-m-added}), 
		\begin{align}
		f_{X^{<i}}(o^*) < \tau_{j-1}. \label{eq:ub-m-not-added}
		\end{align} 
		
		Finally,
		\begin{align*}
			f_{X^{<i}}(x_i) \Geq{Eq~(\ref{eq:ub-m-added})} \tau_j \geq \frac{1}{(1+\eps)} \cdot \tau_{j-1} 
			\Geq{Eq~(\ref{eq:ub-m-not-added})} \frac{1}{(1+\eps)} \cdot f_{X^{<i}}(o^*) 
			\Geq{Eq~(\ref{eq:ub-m-o*})} \frac{f(\OPT) - f(X^{<i})}{(1+\eps) \cdot k},
		\end{align*}
		finishing the proof. 
	\end{proof}
	
	We can now finalize the proof of Lemma~\ref{lem:ub-e-apx} as follows,
	\begin{align*}
		f(\OPT) - f(X) &= f(\OPT) - f(X^{<k}) - f_{X^{<k}}(x_k) \tag{by definition of $f_{X^{<k}}(x_k)$} \\
		&\Leq{Claim~\ref{clm:ub-e-apx-contribution}} f(\OPT) - f(X^{<k}) - \frac{f(\OPT) - f(X^{<k})}{(1+\eps) \cdot k} \\
		&= \paren{1-\frac{1}{(1+\eps)\cdot k}} \cdot \paren{f(\OPT) - f(X^{<k})} \\
		&\leq \paren{1-\frac{1}{(1+\eps)\cdot k}}^{k} \cdot \paren{f(\OPT) - f(X^{<1})} \tag{by applying Claim~\ref{clm:ub-e-apx-contribution} recursively} \\
		&\leq \paren{1/e + \eps} \cdot f(\OPT)
	\end{align*}
	as $f(X^{<1}) = 0$ since $X^{<1} = \emptyset$ by definition. This implies that $f(X) \geq (1-1/e - \eps) \cdot \opt$. 
\end{proof}

Theorem~\ref{thm:dist-upper-e} now follows immediately from Lemma~\ref{lem:ub-e-cc} and Lemma~\ref{lem:ub-e-apx}. 

\medskip

We conclude this section by proving the following corollary of Theorem~\ref{thm:dist-upper-e} for the maximum coverage problem, which formalizes
the first part of Result~\ref{res:dist-upper}. The proof is a direct application of Theorem~\ref{thm:dist-upper-e} plus the known sketching methods for coverage functions in~\cite{McGregorV17,BateniEM17}
to further optimize the communication cost. 

\begin{corollary}\label{cor:dist-upper-e}
	There exists a randomized distributed algorithm for the maximum coverage problem that for any integer $r \geq 1$, and any parameter $\eps \in (0,1)$, with high probability 
	computes an $\paren{\frac{e}{e-1} + \eps}$-approximation in $r$ rounds and $\Ot(\frac{k}{\eps^4} \cdot m^{O(1/\eps \cdot r)} + n)$ total communication. 
\end{corollary}

\newcommand{\VV}{\ensuremath{\mathcal{V}_{\textnormal{\textsf{S}}}}}
\newcommand{\YY}{\ensuremath{\mathcal{Y}}}

\begin{proof}
	Coverage functions are submodular (see Appendix~\ref{app:submodular}) and hence by assigning an item $a_S$ to $V$ for each set $S \in \SS$ in the input collection of the machines, 
	we can directly use the \SPGreedy algorithm for maximum coverage; we only need to communicate the set $S$ as a whole so that the value of $f(\cdot)$ on any set of items for $f$ (i.e., collection of sets in coverage problem) can 
	be computed locally, without any further communication. However, as each set requires (potentially) $\Theta(n)$ bits to represent, the total communication cost of this direct implementation
	is $\Ot(k \cdot m^{O(1/\eps \cdot r)} \cdot n)$, roughly a factor $n$ worse than the bounds in the corollary statement. 
	
	To achieve the bounds in Corollary~\ref{cor:dist-upper-e}, we can use a \emph{sketch} of each set $S \in \SS$ instead of communicating the whole set $S$. 
	In particular, it was shown in~\cite{McGregorV17} (see also~\cite{BateniEM17}) that\footnote{We note the result in~\cite{McGregorV17} works even when the sampling is performed using \emph{limited independence} rather than
	the full independence we state in the Lemma~\ref{lem:mcgregorv17}; however, as we do not need this additional feature, we state the simpler version.},
	
	\begin{lemma}[\!\cite{McGregorV17}]\label{lem:mcgregorv17}
		Suppose $U'$ is a subset of $[n]$ chosen by picking each
		element in $[n]$ independently and w.p. $q = O\paren{\frac{k\log{m}}{\eps^2 \cdot \opt}}$ and define $\SS' = \set{S \cap U \mid S \in \SS}$; then, with high probability, 
		for all collection of $k$ sets $S_1,\ldots,S_k$ in $\SS$ and their corresponding sets $S'_1,\ldots,S'_k$ in $\SS'$, 
		\[\card{S'_1 \cup \ldots \cup S'_k} = \card{S_1 \cup \ldots \cup S_k} \cdot q \pm \eps \cdot q \cdot \opt.\]  
	\end{lemma}
	
	By Lemma~\ref{lem:mcgregorv17}, we can first perform a sampling step to reduce the size of the universe, while ensuring that the returned solution on the subsampled
	universe is still an $\paren{\frac{e}{e-1}+O(\eps)}$ approximation of the original instance. As size of each set $S$ in the original
	instance is clearly $O(\opt)$, after the sampling, the set $\card{S'} = O(\frac{k}{\eps^2}\log{m})$ w.h.p. This step already ensures that we need at most $\Ot(k/\eps^2)$ bits to communicate each set
	as opposed to $\Ot(n)$. 
	
	To shave off another factor of $k$ in the communication, we need to modify $\SPGreedy$ slightly. Firstly, we run $\SPGreedy$ with the parameter $k' = (1-\eps)\cdot k$ instead of the original $k$. 
	Additionally, at each step $t$ in $\SPGreedy$, the machines compute the sample collection of sets to send to the coordinator as before. Let us denote this
	collection by $\VV^{t}$ ($\textnormal{\textsf{S}}$ stands for sample here). Instead of sending $\VV^{t}$ to the coordinator directly, the machines first sample a collection of 
	$\eps \cdot \frac{k}{2r}$ uniformly at random chosen sets from $\VV^{t}$ and communicate them to the coordinator.  The coordinator adds these sets to a collection $\YY$ (maintained throughout the algorithm
	similar to the partial solution $\XX$, corresponding to set of items $X$ in the submodular maximization notation), and communicates back the elements in these sets to each machine. 
	The machines then remove any element from the universe that is covered by these sets; then, they send the collection $\VV^{t}$ in $\SPGreedy$ (after removing the mentioned elements) to the coordinator. 
	The rest of the protocol is exactly as before. At the end, the coordinator outputs $\YY \cup \XX$ as the solution. 
	Notice that this change doubles the number of rounds in $\SPGreedy$ as each step now requires two rounds of communication. 
	
	The correctness of the algorithm follows exactly as before since $\YY$ can contain at most $\eps \cdot k$ sets and $\XX$ has size $k' = (1-\eps)\cdot k$ and hence $\YY \cup \XX$ is a valid $k$-cover of 
	the universe. Moreover, the same exact argument in Theorem~\ref{thm:dist-upper-e} ensures that  $\XX$ covers at least $(1-1/e - \eps)$ of the optimal $k'$-cover on $[n] \setminus c(\YY)$, and hence 
	is a $\paren{\frac{e}{e-1}+O(\eps)}$-approximate $k$-cover also over $[n] \setminus c(\YY)$ (as the best $k'$-cover for $k' = (1-\eps) \cdot k$ is a $(1+\eps)$ approximation of best $k$-cover). Since anything in
	$c(\YY)$ is covered by $\YY$, the algorithm achieves a $\paren{\frac{e}{e-1}+O(\eps)}$-approximation. 
	
	Finally,  we bound the communication cost of this algorithm. After sending the additional samples from each $\VV^{t}$, we know that 
	any element in the universe that appears in more than $(r/\eps k) \cdot O(\log{m})$ fraction of the candidate sets in $\VV^{t}$ is being covered by $\YY$ with high probability and hence 
	after removing $c(\YY)$ from the universe, size of each set in $\VV^{t}$ is now in average only $O(r/\eps^3 \cdot \log{m}) = \Ot(1/\eps^4)$ (since
	$r = O(\log{m}/\eps)$, as we never need to run the algorithm for more than that many rounds). This means that the total number of bits needed to communicate $\VV^{t}$ is now  $\Ot(1/\eps^4) \cdot \card{\VV^{t}}$
	bits, which finalizes the proof. 
\end{proof}

\section{Applications to Other Models of Computation}\label{sec:applications}

We discuss the applications of our results to maximum coverage (and submodular maximization) in the dynamic streaming model and the MapReduce framework introduced in Section~\ref{sec:results}. 
We finish the section by making a remark about the role of partitioning of the input in the distributed model. 

\subsection{Maximum Coverage in Dynamic Set Streams}\label{sec:dynamic}

We first define the dynamic set streaming model formally. The definition is a straightforward extension of the set streaming model
introduced by Saha and Getoor~\cite{SahaG09} (see also~\cite{EmekR14}) to dynamic streams similar to dynamic graph streams~\cite{AhnGM12}. 
Indeed, if we consider the maximum coverage as a \emph{hypergraph} problem, i.e., picking $k$ hyperedges to cover the most number of vertices (similar to~\cite{EmekR14} for streaming set cover),
then our notion of dynamic set streams is exactly the same as \emph{dynamic hypergraph streams} in~\cite{GuhaMT15}. 

\begin{definition}\label{def:dynamic-stream}
	A dynamic set stream $\FF_n = \langle a_1, a_2, \ldots, a_t \rangle$ defines a set-system $\SS$ over $[n]$. Each $a_i$ is a tuple $(S_i,\Delta_i)$ where $S_i \subseteq [n]$ and $\Delta_i \in \set{-1,+1}$. 
	The multiplicity of a set $S \subseteq [n]$ is defined as: 
	\begin{align*}
		\FF(S) := \sum_{a_i: S_i = S} \Delta_i.
	\end{align*}
	The multiplicity of every set is required to be always non-negative during the stream. We use the $(2^{n})$-dimensional vector $f$ to denote the vector of multiplicities of the sets seen in the stream.  
\end{definition}

All known algorithms for all problems in dynamic streams (not only dynamic set streams) have a similar
form: they first choose a (possibly random) integer matrix $A$ and maintain the \emph{linear sketch} $A \cdot f$ in the stream. At the end 
of the stream, they use $A \cdot f$ to compute the answer. It was shown by~\cite{LiNW14} that this is not a coincidence; any one pass streaming algorithm for approximating
any arbitrary function on multiplicity vector $f$ in the dynamic streaming model can be reduced to an algorithm which, before the stream begins, samples a matrix $A$ uniformly at random from a set of hardwired integer matrices, and then maintains the linear sketch $A \cdot f \mod q$, where $q = (q_1,\ldots,q_r)$ is a vector of positive integers and $r$ is the number of rows of $A$. The space complexity of this linear sketching algorithm is only 
larger by an additive factor of the space required to sample $A$ and $q$ (which is shown to be logarithmic in the dimension of the vector $f$ in~\cite{LiNW14}). 
This reduction was further extended by~\cite{AiHLW16} to algorithms which make any number of passes, showing the optimal algorithm is to \emph{adaptively} choose a new linear sketch at the beginning of
each pass based on the computation in previous passes. 

It is a well-known fact that any linear sketching algorithm that requires at most $p$ passes of adaptive sketching can be implemented in the communication model 
studied in this paper (see Appendix~\ref{app:communication}) with $p$ rounds of communication: each player simply computes
 the linear sketches on its input and writes that on the shared blackboard; by linearity of the sketches, the players
can then combine these sketches and obtain a linear sketch of the whole input. This allows the players to implement each round of adaptive sketching in one round of communication and compute the final answer. 
It is also easy to see that the \emph{per player} communication cost of this new algorithm is at most the size of the linear sketch. 
Combining this with the reduction of~\cite{AiHLW16} implies that if one can prove a lower bound on the per player communication complexity of a problem in the shared blackboard model, 
one also obtains a lower bound on the space complexity of dynamic streaming algorithms; notice that since in the communication model we can perform the sampling of $A$ and $q$ via public
randomness, free of communication charge, we do not even need to pay for the extra additive factor in space in the reduction; we refer the interested reader to~\cite{AiHLW16} for more details. 

The takeaway is that by applying the reduction of~\cite{AiHLW16} to dynamic streaming algorithms for maximum coverage problem and using our
lower bound in Theorem~\ref{thm:dist-lower} (which was proven in this more general communication model), we obtain that, 

\begin{corollary}\label{cor:dynamic-lower}
	No $p$-pass semi-streaming algorithm for the maximum coverage problem in the dynamic streaming model can approximate the value of optimal solution to a factor of $o(\frac{k^{1/2p}}{p\cdot\log{k}})$ with a sufficiently 
	large probability. 
\end{corollary}

We remark that one can obtain the same exact bounds in Theorem~\ref{thm:dist-lower} for the space complexity of dynamic streaming algorithms also; however, as our focus is on semi-streaming algorithms we provide the 
above theorem which is qualitatively similar but is easier to parse. 

We now turn to proving an upper bound for maximum coverage in dynamic streams using Theorem~\ref{thm:dist-upper-e}. We remark that the same argument holds also for maximizing any monotone submodular function subject 
to a cardinality constraint (exactly as in Section~\ref{sec:dist-m}); for brevity, in the following we only focus on the maximum coverage problem. 

We show that \SPGreedy can be implemented in dynamic streams. 
To do this, we need a primitive that allows for sampling a set from a dynamic stream uniformly at random. This can be achieved using $\ell_0$-samplers introduced in~\cite{FrahlingIS2008}. Since the dimension
of the multiplicity vector $f$ is $2^{n}$ and each set also requires $\Theta(n)$ bits to represent, a naive implementation of the best known streaming $\ell_0$-samplers due to~\cite{JowhariST2011} requires
$\Theta(n^2)$ space. However, using the fact there can only be $m$ non-zero entries in the vector $f$ at the end of the stream (as number of sets is at most $m$), we can implement the algorithm of~\cite{JowhariST2011}
with only $O(n \cdot \poly\set{\log{m},\log{n}})$ space (simply change the number of buckets in Theorem~2 in~\cite{JowhariST2011} from $n$ to $\Theta(\log{m})$). We refer to this primitive as a \emph{set sampler}.

Having the set sampler primitive; it is now easy to see that we can implement the \SPGreedy algorithm in dynamic stream. Each of the $s$ steps of $\SPGreedy$ (in any iteration) can be implemented
by making one pass over the stream and maintaining a set sampler over the collection of sets defined in Line~(\ref{line:spgreedy-set-defined}) of \SPGreedy; notice that whenever a set is updated in the stream we can decide
in $\Ot(n)$ space whether it belongs to this collection or not and hence send it to the set sampler primitive. The rest of the algorithm is exactly as in \SPGreedy and its modification in Corollary~\ref{cor:dist-upper-e}. 
By running \SPGreedy with $p = O(\log{m}/\eps)$ passes over the stream, we obtain an algorithm with space complexity of $\Ot(k/\eps^4 \cdot m^{O(1/\eps \cdot p)} + n) = \Ot(k/\eps^4 + n)$, i.e., a semi-streaming algorithm. 
Consequently, 

\begin{corollary}\label{cor:dynamic-upper}
	There exists a randomized semi-streaming algorithm for the maximum coverage problem that for any constant $\eps \in (0,1)$, with high probability, 
	computes an $\paren{\frac{e}{e-1} + \eps}$-approximation in $O(\log{m}/\eps)$ passes over the stream. 
\end{corollary}

Corollary~\ref{cor:dynamic-upper} can also can be stated for dynamic streaming algorithms with different space bounds corresponding to Corollary~\ref{cor:dist-upper-e}; however, for brevity, we 
only focused on semi-streaming algorithms. 

\paragraph{Constant Pass Algorithms.} We remark that our second algorithm in Result~\ref{res:dist-upper} does not admit a linear sketching implementation; in fact, using Corollary~\ref{cor:dynamic-lower}, it is easy to prove that the \GreedySketch subroutine used by each machine \emph{cannot} be implemented in dynamic streams in less than logarithmic number of passes over the stream. 
As a result, we do not know if one can achieve a {non-trivial} semi-streaming algorithm for maximum coverage
in dynamic streams in \emph{constant passes} over the stream. In particular, can we match the lower bound in Corollary~\ref{cor:dynamic-lower} for any number of passes $p$?
We leave this as an intriguing open question. 

\paragraph{Application to Set Cover in Dynamic Streams.} 
We finish this section by stating that our algorithm in Corollary~\ref{cor:dynamic-upper} can also be used to obtain the first dynamic streaming algorithm for the set cover problem. The algorithm is as follows. Guess the value of optimal solution $\optp$ for set cover in powers of two in parallel and perform the following procedure.  Run the algorithm in Corollary~\ref{cor:dynamic-upper} 
for each guess with the parameter $k = \optp$ and remove all covered elements from the universe; repeat this process until there is no uncovered element left; return the collection of all sets computed over different passes
as a set cover. It is easy to see that $O(\log{n})$ iteration of this algorithm suffices to cover all the elements (hence the $O(\log{n})$ factor in the approximation ratio) 
and each iteration can be implemented in $O(\log{m})$ passes by Corollary~\ref{cor:dynamic-upper} ($O(\log{m}\cdot\log{n})$ passes in total). 
As a result, 

\begin{corollary}\label{cor:set-cover}
	There exists a randomized semi-streaming algorithm for the set cover problem that with high probability computes an $O(\log{n})$-approximation 
	in $O(\log{m}\cdot\log{n})$ passes over the stream. 
\end{corollary}

\subsection{Maximum Coverage in the MapReduce Framework}\label{sec:map-reduce}

We now present our results for maximum coverage and submodular maximization in the MapReduce framework described in Section~\ref{sec:results}. 

Recall that in the sketch-and-update approach (described in Section~\ref{sec:results}) in the MapReduce framework, in each round, every machine is sending a message directly to a designated central machine for combining the 
sketches. By definition of the MapReduce framework, the total messages received by the central machine can only be proportional to its memory which is of size $O(s)$. This enforces an upper bound on the total 
communication of $O(s)$ in each round by the machines. It is thus easy to see that efficient MapReduce algorithms in the sketch-and-update framework immediately imply communication efficient protocols in the distributed 
coordinator model (note that this is in general is not true for every MapReduce algorithm). As a result, we can interpret Theorem~\ref{thm:dist-lower} as proving a lower bound for sketch-and-update algorithms in the MapReduce 
framework. 

\begin{corollary}\label{cor:mapreduce-lower}
	For any $\delta \in (0,1)$, any MapReduce algorithm in the sketch-and-update framework described in Section~\ref{sec:results} that uses $s = m^{\delta}$ space per machine 
	and computes a constant factor approximation to maximum coverage requires $\Omega(\frac{1}{\delta})$ rounds of communication. 
\end{corollary}

Moreover, both algorithms in Result~\ref{res:dist-upper} can be implemented in the MapReduce model. In particular, we state the following corollary of Theorem~\ref{thm:dist-upper-e} for submodular maximization
which also subsumes the results for coverage maximization. 

\begin{corollary}\label{cor:mapreduce-upper}
	Let $V$ be a universe of $m$ items and $f : 2^{V} \rightarrow \IR^{+}$ be a monotone submodular function. 
	For any $\eps,\delta \in (0,1)$, there exists an $\paren{\frac{e}{e-1} + \eps}$-approximation randomized algorithm for maximizing $f$ subject to a cardinality constraint in the MapReduce framework that
	uses $p = O(m^{1-\delta/\eps})$ machines each with $s = O(m^{\delta/\eps})$ memory and computes the answer in $O(\frac{1}{\eps \cdot \delta})$ rounds. 
\end{corollary}

As stated in Section~\ref{sec:results}, our bounds in Corollary~\ref{cor:mapreduce-upper} matches the best known bounds of~\cite{BarbosaENW16} with the additional benefit of having sublinear in $m$ communication. 
We again remark that the algorithm in~\cite{BarbosaENW16} is however more general in that it supports a larger family of constraints beside the cardinality constraint we studied in this paper.

\subsection{Adversarial vs Random Partitions}\label{sec:partition} 

We considered \emph{adversarial} input partitions in this work, meaning that the input across the machines is distributed adversarially. 
Several recent work have studied optimization problems in the distributed model when the input is \emph{randomly} partitioned~\cite{MirrokniZ15,BarbosaENW15,AssadiK17}. 
For maximum coverage (and submodular maximization), it was shown previously that under this assumption one can achieve a constant factor approximation using $\Ot(n)$ communication per machine
in only \emph{one} round of communication~\cite{MirrokniZ15,BarbosaENW15}. Comparing this with Theorem~\ref{thm:dist-lower} implies that an approximation factor that can
 be achieved in only one round of communication and $\Ot(n)$ 
communication under randomized partitions, \emph{cannot} be achieved in $o(\frac{\log{n}}{\log\log{n}})$ rounds of communication and $\poly(n)$ communication in adversarial partitions! 

We remark that separations on the round complexity of 
randomized and adversarial partitions were known for some problems before (see e.g.,~\cite{MunroP80,GuhaM09,ChakrabartiJP08} for median estimating). The striking gap between these
two cases for the distributed maximum coverage problem is another nice illustration of this phenomenon.


\subsection*{Acknowledgements} 
The first author is grateful to Alessandro Epasto for bringing~\cite{EpastoLVZ17} to his attention and to David Woodruff for 
a helpful discussion on the implication of the results in~\cite{AiHLW16} for proving multi-pass dynamic streaming lower bounds. 
We also thank Paul Liu and Jan Vondrak for helpful comments on the presentation of the paper.

\bibliographystyle{abbrv}
\bibliography{general}

\clearpage

\appendix

\newcommand{\rA}{\rv{A}}
\newcommand{\rB}{\rv{B}}
\newcommand{\rC}{\rv{C}}
\newcommand{\rD}{\rv{D}}

\section{Tools From Information Theory}\label{app:info}

The proof of the following basic properties of entropy and mutual information
can be found in~\cite{ITbook} (see Chapter~2).

\begin{fact}\label{fact:it-facts}
  Let $\rA$, $\rB$, and $\rC$ be three (possibly correlated) random variables.
   \begin{enumerate}
  \item \label{part:uniform} $0 \leq \HH(\rA) \leq \card{\rA}$, and $\HH(\rA) = \card{\rA}$
    iff $ A$ is uniformly distributed over its support.
  \item \label{part:info-zero} $\mi{\rA}{\rB \mid \rC} \geq 0$. The equality holds iff $\rA$ and
    $\rB$ are \emph{independent} conditioned on $\rC$. 
  \item \label{part:cond-reduce} $\HH(\rA \mid  \rB, \rC) \leq \HH(\rA \mid \rB)$.  The equality holds iff $\rA \perp \rC \mid \rB$.
  \item \label{part:chain-rule} $\mi{\rA, \rB}{\rC} = \mi{\rA}{\rC} + \mi{\rB}{\rC \mid \rA}$ (\emph{chain rule of mutual information}).
  \item \label{part:data-processing} Suppose $f(\rA)$ is a deterministic function of $\rA$, then $\mi{f(\rA)}{\rB \mid \rC} \leq \mi{\rA}{\rB \mid \rC}$ (\emph{data processing inequality}). 
   \end{enumerate}
\end{fact}

We also use the following two standard propositions, regarding the effect of conditioning on mutual information.

\begin{proposition}\label{prop:info-increase}
  For variables $\rA, \rB, \rC, \rD$, if $\rA \perp \rD \mid \rC$, then, $\mi{\rA}{\rB \mid \rC} \leq \mi{\rA}{\rB \mid  \rC,  \rD}$.
\end{proposition}
 \begin{proof}
  Since $\rA$ and $\rD$ are independent conditioned on $\rC$, by
  \itfacts{cond-reduce}, $\HH(\rA \mid  \rC) = \HH(\rA \mid \rC, \rD)$ and $\HH(\rA \mid  \rC, \rB) \ge \HH(\rA \mid  \rC, \rB, \rD)$.  We have,
	 \begin{align*}
	  \mi{\rA}{\rB \mid  \rC} &= \HH(\rA \mid \rC) - \HH(\rA \mid \rC, \rB) = \HH(\rA \mid  \rC, \rD) - \HH(\rA \mid \rC, \rB) \\
	  &\leq \HH(\rA \mid \rC, \rD) - \HH(\rA \mid \rC, \rB, \rD) = \mi{\rA}{\rB \mid \rC, \rD}. \qed
	\end{align*}
	
\end{proof}

\begin{proposition}\label{prop:info-decrease}
  For variables $\rA, \rB, \rC,\rD$, if $ \rA \perp \rD \mid \rB,\rC$, then, $\mi{\rA}{\rB \mid \rC} \geq \mi{\rA}{\rB \mid \rC, \rD}$.
\end{proposition}
 \begin{proof}
 Since $\rA \perp \rD \mid \rB,\rC$, by \itfacts{cond-reduce}, $\HH(\rA \mid \rB,\rC) = \HH(\rA \mid \rB,\rC,\rD)$. Moreover, since conditioning can only reduce the entropy (again by \itfacts{cond-reduce}), 
  \begin{align*}
 	\mi{\rA}{\rB \mid  \rC} &= \HH(\rA \mid \rC) - \HH(\rA \mid \rB,\rC) \geq \HH(\rA \mid \rD,\rC) - \HH(\rA \mid \rB,\rC) \\
	&= \HH(\rA \mid \rD,\rC) - \HH(\rA \mid \rB,\rC,\rD) = \mi{\rA}{\rB \mid \rC,\rD}. \qed
 \end{align*}
 
\end{proof}

For two distributions $\mu$ and $\nu$ over the same probability space, the \emph{Kullback-Leibler divergence} between $\mu$ and $\nu$ is defined as $\DD{\mu}{\nu}:= \Ex_{a \sim \mu}\Bracket{\log\frac{\Pr_\mu(a)}{\Pr_{\nu}(a)}}$.
We have,
\begin{fact}\label{fact:kl-info}
	For random variables $\rA,\rB,\rC$, 
	\[\mi{\rA}{\rB \mid \rC} = \Ex_{(b,c) \sim \distribution{\rB,\rC}}\Bracket{ \DD{\distribution{\rA \mid \rC=c}}{\distribution{\rA \mid \rB=b,\rC=c}}}.\] 
\end{fact}

We denote the \emph{total variation distance} between two distributions $\mu$ and $\nu$ over the same probability space $\Omega$ by $\tvd{\mu}{\nu} = \frac{1}{2} \cdot \sum_{x \in \Omega} \card{\Pr_{\mu}(x) - \Pr_{\nu}(x)}$. 

The following Pinskers' inequality bounds the total variation distance between two distributions based on their KL-divergence, 

\begin{fact}[Pinsker's inequality] \label{fact:pinskers}
	For any two distributions $\mu$ and $\nu$, $\tvd{\mu}{\nu} \leq \sqrt{\frac{1}{2} \cdot \DD{\mu}{\nu}}$. 
\end{fact}

Finally, 
\begin{fact}\label{fact:tvd-small}
	Suppose $\mu$ and $\nu$ are two distributions for an event $\event$, then, $\Pr_{\mu}(\event) \leq \Pr_{\nu}(\event) + \tvd{\mu}{\nu}$. 
\end{fact}

\section{Further Discussion on Our Framework}\label{app:discussions}

We discuss further extensions to our framework for proving communication lower bounds for bounded round protocols  
introduced in Section~\ref{sec:dgl} including how to use the framework to obtain lower bounds for search problems and the connection of this framework to
previous results in~\cite{DobzinskiNO14,Konrad15,AssadiKLY16,AssadiKL17,AlonNRW15,Assadi17ca}. 

\paragraph{Search Problems.} We can also use our framework to prove a lower bound for a search problem $\PP : \set{0,1}^{s} \mapsto R_s$ (for some range $R_s$); for example, think of $\PP$ as finding edges of an approximate 
matching. The framework is as before for the most part. In a search problem, we do not have \Yes and \No instances, rather all instances are sampled from the same distribution, and the goal of the players is to find a 
suitable answer in $R_s$, e.g., a large matching in the example above. In the following, we discuss the changes needed in our framework to be able to prove lower bounds for search problems as well.

Instead of having $g_r$ copies of the same special instance, we sample the special instance 
of each group \emph{independently} from $\dist_{r}$ (note that this is \emph{not} possible for a decision problem because we need all special instances to be either a \Yes instance or a \No instance which correlates them).
We also change the definition of $\gamma$-preserving property slightly so that it ensures that to solve $\PP_{s_r}$, \emph{at least one group} of players need to solve $\PP_{s_{r-1}}$ on their special
instance w.p. $1-\gamma$. With this property, we can do the embedding of $(r-1)$-round instances in $r$-round instances in Lemma~\ref{lem:dgl-round-elimination} as before with a slight change; the players in $\PP_{s_{r-1}}$ only 
need to embed their input in one group of the distribution $\dist_r$ (as opposed to ``copying'' themselves $g_r$ times) and can sample the input for rest of the groups using public randomness. 

The proof is similar as before with one crucial change. Since special instances are sampled independently, one can in fact show a stronger result than the one in Lemma~\ref{lem:dgl-info-special} for decision problems; in particular, 
for any group $i \in [g_r]$, one can now show that, 
\begin{align*}
	\mi{I^{i}_{\jstar}}{\rProt_1 \mid \rPhi,\rJ} \leq \sum_{q \in P_i} {\card{\rProt_{1,q}}}/{w_r}. 
\end{align*}
In other words, as only the players in group $P_i$ can communicate information about the special instance of this group and hence the information revealed about this special instance is bounded by the message length of \emph{this} 
particular group, and not all players. This improves the per player communication lower bound by a factor of $g_r$ in each round, which is crucial for some application, e.g., in~\cite{AlonNRW15}. 
The rest of the proof is as before. 

\paragraph{Connection to Previous Work.} The framework introduced in Section~\ref{sec:dgl} plus the extension for search problems
subsumes the communication lower bounds in~\cite{AlonNRW15,Assadi17ca}, the lower bound for super constant estimation algorithms of matching size in \emph{dense} graphs in~\cite{AssadiKL17},
and the lower bound for combinatorial auctions in~\cite{DobzinskiNO14}. We again emphasize that our framework only facilitates proving the communication lower bound in those arguments; to obtain the desired 
bound on the approximation ratio subject to this communication lower bound, one still needs to instantiate the framework with suitable packing and labeling functions that are designed specifically for the problem at hand
at in each of these results separately. 

To obtain the lower bounds in~\cite{Konrad15,AssadiKLY16}, the $(1+\eps)$-approximation lower bounds in~\cite{AssadiKL17}, and the lower bound for unit-demand auctions (matching markets)
in~\cite{DobzinskiNO14}, we need to modify the framework as follows: in our framework, we use the fact that a protocol that reveals $o(1)$ bits of information in distribution $\dist_0$, cannot solve 
the problem $\PP_0$ with probability more than $1/2+o(1)$, which is always true (by Lemma~\ref{lem:dgl-istar-close} and Fact~\ref{fact:tvd-small}). 
However, in some scenarios, if we insists on only revealing $o(1)$ bits of information about $\dist_0$, we cannot hope to achieve
any meaningful lower bounds for $\dist_1$ (as $w_1$ cannot be sufficiently large in the parameters of the problem); this is the case for the aforementioned results. 
To achieve those (and similar) simultaneous lower bounds, we need to apply Lemma~\ref{lem:dgl-info-special} with larger values of $\card{\Prot_1}$ and obtain that, for some suitably chosen value of $t$, only $o(t)$ bits of information
are revealed about the instance of $\dist_0$; this allows for proving a communication lower bound of $\Omega(w_1 \cdot t)$ instead of  $\Omega(w_1)$ which follows directly from our framework. 
However, in this case, one needs to also argue that revealing $o(t)$ bits of information about $\dist_0$ still does not allow for solving this problem with a sufficiently large probability. 
This step is again problem specific and was shown to be correct for~\cite{Konrad15,AssadiKLY16} using a combinatorial argument, in~\cite{DobzinskiNO14} using a reduction to set disjointness
in communication complexity~\cite{Razborov92}, and in~\cite{AssadiKL17} using an information complexity lower bound for the boolean hidden hypermatching problem~\cite{VerbinY11}. 
We remark that except for the argument in~\cite{AssadiKL17}, the aforementioned results were
proven using different combinatorial arguments; our framework suggests a unified approach for proving all these lower bounds.

\end{document}